\documentclass[12pt,onecolumn]{IEEEtran}

\makeatletter
\long\def\@makecaption#1#2{\ifx\@captype\@IEEEtablestring%
	\footnotesize\begin{center}{\normalfont\footnotesize #1}\\
		{\normalfont\footnotesize\scshape #2}\end{center}%
	\@IEEEtablecaptionsepspace
	\else
	\@IEEEfigurecaptionsepspace
	\setbox\@tempboxa\hbox{\normalfont\footnotesize {#1.}~~ #2}%
	\ifdim \wd\@tempboxa >\hsize%
	\setbox\@tempboxa\hbox{\normalfont\footnotesize {#1.}~~ }%
	\parbox[t]{\hsize}{\normalfont\footnotesize \noindent\unhbox\@tempboxa#2}%
	\else
	\hbox to\hsize{\normalfont\footnotesize\hfil\box\@tempboxa\hfil}\fi\fi}
\makeatother
\usepackage{algorithm2e,setspace}

\usepackage{amsmath}

\usepackage{amsthm}
\usepackage{hyperref}
\newtheorem{remark}{Remark}

\usepackage{subfigure}
\usepackage{amsfonts}
\usepackage{epsfig}
\usepackage{amssymb}
\usepackage{amsmath}
\usepackage{cite}
\usepackage{color,soul}
\usepackage{subfigure}
\usepackage{multirow}
\usepackage{rotating}
\usepackage{graphicx}
\usepackage{tabularx}
\usepackage{array}
\usepackage{color,soul}
\usepackage{graphicx,dblfloatfix}
\usepackage{blindtext}
\usepackage{dsfont}

\usepackage[T1]{fontenc}                 
\usepackage[utf8]{inputenc}              

\newtheorem{thm}{Theorem}
\newtheorem{lem}[thm]{Lemma}

\ifCLASSINFOpdf
\else
\fi

\begin{document}
%
\title{Compound Poisson  Noise Sources in Diffusion-based Molecular Communication}

\author{Ali Etemadi$^{\dagger}$, Paeiz Azmi$^{\dagger}$, Hamidreza~Arjmandi$^*$, Nader Mokari$^{\dagger}$\\
$^{\dagger}$Tarbiat Modares University, $^*$Yazd University
}

\maketitle

\begin{abstract}
Diffusion-based molecular communication (DMC) is one of the most promising approaches for realizing nano-scale communications for healthcare applications. The DMC systems in in-vivo environments may encounter biological entities that release molecules identical to the molecules used for signaling as part of their functionality. Such entities in the environment act as external noise sources  from the DMC system's perspective. In this paper, the release of molecules by biological external noise sources is particularly modeled as a compound Poisson process.  The impact of compound Poisson noise sources (CPNSs) on the performance of a point-to-point DMC system is investigated. To this end, the noise from the CPNS observed at the receiver is characterized. Considering a simple on-off keying (OOK)  modulation and formulating  symbol-by-symbol  maximum likelihood (ML) detector, the performance of DMC system in the presence of the CPNS is analyzed. For special case of CPNS in high-rate regime, the noise received from the CPNS is approximated as a Poisson process whose rate is normally  distributed. In this case, it is proved that a simple single-threshold detector (STD) is an optimal ML detector. Our results reveal that in general, adopting the conventional simple homogeneous Poisson noise model may lead to overly optimistic performance predictions, if a CPNS is present.
      

\end{abstract}


%
\IEEEpeerreviewmaketitle

\section{Introduction}
%
%
%
%

\subsection{Motivation}
Diffusion based molecular communication (DMC) is a promising approach for realizing nano communications \cite{Akyldiz2011}. In DMC, information is encoded in the concentration, type, and/or release time of  molecules. In particular, a transmitter nanomachine releases  information molecules  into the environment. The released molecules move randomly via Brownian motion and, as a consequence, some molecules may be observed (received) at the receiver \cite{Pierobon2011}.
Specific features of DMC, such as its bio-compatibility, make it attractive for  healthcare applications 
\cite{Akyldiz2015},\cite{nanomedicine},\cite{Chahibi2013}.  
However, the application of DMC systems in in-vivo environments faces many practical challenges and requires extensive research and development. Particularly, it is essential to analyze the impact of the biological external noise sourses on the performance of DMC systems \cite{a5}.
\color{black}

The biological entities in the body release different types of molecules as part of their functionality. The molecule release processes of biological entities exhibit a random behavior both for the time of release and the number of the released molecules. For instance, Poisson, Gaussian, and Weibull renewal models have been proposed for the timing of the bursts in the endocrine systems \cite{hormone}. In particular, for the neuroendocrine system, the secretory bursts at random time instants have been modeled as a Poisson point process \cite{Neuro},\cite{Neuro2}. 
As another example, the release of neurotransmitters  in synapses has been characterized by a doubly stochastic Poisson process \cite{Synapse2}.
Also, the numbers of molecules transported by ion channels and ion pumps  across the cell membrane can be stochastically modeled as Poisson random variables (RVs) \cite{Pump}. DMC systems operating in in-vivo environments may encounter biological entities that release molecules identical to the molecules  used for signaling. Such entities in the environment act as external noise sources from the DMC system's perspective.  Therefore, accurate  modeling of such noise sources taking into account their intrinsic characteristics is crucial for comprehensive performance analysis and evaluation. 
\subsection{Related Works}
\definecolor{alizarin}{rgb}{0.82, 0.1, 0.26}
In the MC literature, different noise models have been proposed to characterize the uncertainty inherent to the molecule release process at the transmitter, Brownian motion, the reception process at the receiver, and environment noise.
In \cite{Noise2011}, the noises introduced by the transmitter and the diffusion channel are  modeled as  additive Gaussian noise and are referred to as particle sampling and particle counting noise, respectively. In \cite{eckford}, it is shown that additive inverse Gaussian noise can be used to model the molecular timing channel where the information is encoded into the release time of the molecules into the fluid medium with drift. In \cite{Arjmandi2013}, the authors propose a Poisson model to characterize the noise due to the randomness of the transmitter release process and the Brownian motion in the fluid medium. 
Also, in \cite{nariman}, additive stable distribution noise is introduced to characterize molecular timing channel for different modulation schemes. The authors in \cite{mallik} consider the continuous collision of molecules as source of noise leading to uncertainty in the position of the molecules. To mathematically model this noise source, the Langevin model for Brownian motion in a fluid medium is considered. 
The uncertainty caused by the reception process of ligand receptors is considered in  \cite{arash}-\cite{aminian}. 
Unlike for the noise introduced by the transmitter release, Brownian motion, and the reception process, less considerations has been given to the environmental noise in the DMC systems. The authors of \cite{Arjmandi2013} model environmental noise by a homogeneous Poisson distribution whose parameter is equal to the average number of molecules received during a given time slot. The homogeneous Poisson noise model is more elaborately presented  in \cite{Mosayebi2014} where the dependence of the average number of noise molecules, i.e., the parameter of the distribution, on the time-slot duration is taken into account. 
However, these conventional noise models which are homogeneous in space and time are not capable of accurately modeling the noise introduced by biological entities, which is the main focus of this paper. 

 Authors in \cite{a5} analyze the expected number of molecules observed at the transparent receiver originated from the external noise sources, e.g., multiuser interference caused by the transmitters of other communication links, unintended leakage from vesicles, or the output from an unrelated biochemical process. In this work, the statistics of the emission process of the external noise source is neglected by assuming a uniform emission process and the proposed analysis focuses on the expected impact and not the complete probability density function of the impact of external noise source.


\subsection{Proposed Model}

Considering a biological external noise source releasing random numbers of molecules in random time instances leads to receiver noise whose statistics differ from that of the conventional noise models. In other words, conventional noise models are not able to model the noise caused by the biological external noise source. 
The release process of the biological external noise source is particularly modeled as a compound Poisson process (CPP) where the release time events constitute  points of a Poisson process and the amplitudes of  the events (the number of released molecules at a release time event) are random. This model is inspired from the molecule release processes observed in some endocrine systems \cite{hormone}-\cite{Neuro2} and the release of neurotransmitters in synapses \cite{Synapse2} which is the fundamental process that drives information transfer between the neurons in the nervous system. Noteworthy, the CPP is unable to model all biological release processes. At least, Weibull renewal process is a more general model than the CPP, for the release processes in the endocrine systemes \cite{hormone}. 
A biological external noise source whose release process is a CPP is referred to as a compound Poisson noise source (CPNS) in the following. The CPNS takes into account the randomness of the release events of biological noise source in both time and amplitude.  In this paper, we provide a framework for analyzing the performance of DMC systems in the presence of the CPNS. 

A point-to-point DMC system is considered in the presence of a CPNS which releases molecules of the same type as the signaling molecules. 
To investigate the impact of  CPNS on communication performance, a point-to-point DMC system with simple on-off keying modulation is considered. The number of noise molecules observed at the receiver is analyzed and approximated by a Poisson mixture distribution by using the rare-event property of Poisson distributions. Considering a simple on-off keying modulation, the performance of DMC system in the presence of CPNS is analyzed. A symbol-by-symbol maximum likelihood (ML) detector is derived and bit error rate (BER) is obtained by adopting a simple single-threshold detector (STD). For the special case of CPNS in high-rate regime (high rates of release time events), the noise received from the CPNS is approximated as a Poisson distribution whose mean is  normally  distributed. It is proved that the noise received from the CPNS in high-rate regime has a log-concave distribution and a STD is  the  optimal ML detector. For the general case of CPNS, our results report that the distribution of noise received from the CPNS may not be log-concave and can even be multimodal distribution leading to optimality of multiple-threshold detector (not a simple STD).   
Moreover, our particle based simulation (PBS) results confirm the obtained analytical BER expressions. Also, it is revealed that the conventional homogeneous Poisson noise model is not applicable to CPNSs and leads to overly optimistic performance estimates.


The remainder of this paper is organized as follows: In Section \ref{section2}, we present the DMC system model including the  transmitter, receiver, channel, and  CPNS models. In Section \ref{section3}, the distributions of the received signals due to the release of molecules by the transmitter and the CPNS are derived. The optimal ML detector and the error probability of the DMC system in the presence of a CPNS are analyzed in  Section \ref{section4}. In Section \ref{section5}, we provide simulation and numerical results. Finally, the paper is concluded in Section \ref{section6}.

\section{System Model} \label{system model} \label{section2}

A point-to-point DMC system is considered in the presence of a CPNS. Considering the main focus of this paper which is to characterize the effect of CPNS on the DMC, we adopt simplifying assumptions on DMC system and environment. We assume a 3-dimentioanl unbounded environment where the CPNS and transmitter are point sources and the receiver is transparent \cite{a5}. 
 It is assumed that the receiver is at the origin of the coordinate system and is synchronized with the transmitter \cite{syncArj}. The transmitter and the CPNS
 are  located at distances of $d_T$ and $d_C$ from the receiver, respectively; see Fig. \ref{fig:system model}. The signaling molecules and the noise molecules released by the transmitter and the CPNS, respectively, are of the same type  $A$. Simple on-off keying modulation with time-slot duration $T$ is adopted where bits 1 and 0 are represented by the release of $N$ molecules (on average) and no molecule at the beginning of  each time slot, respectively. Assuming  transmission of bit 1, the number of molecules released by the transmitter follows  a Poisson distribution with mean $N$ \cite{Arjmandi2013}. The receiver is assumed to be a transparent spherical volume of radius $r_R$ that counts the number of molecules inside the receiver volume  at sampling time $t_s$ \cite{schober2015}. The receiver uses the observed sample to decide about the  transmitted bit. In the rest of this section, the  CPNS model is presented and the adopted channel model is described.

\subsection{CPNS Model}

\begin{figure}
	\center
	\includegraphics[width=10 cm,height=6 cm]{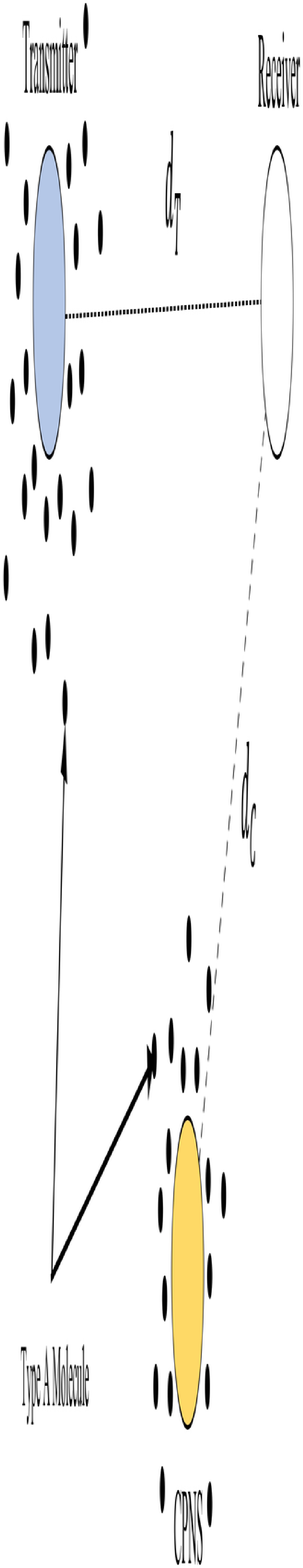}
	\setlength{\abovecaptionskip}{0 cm}
	\caption{Point-to-point DMC system in the presence of a CPNS.}
	\label{fig:system model}
\end{figure}


In this subsection, we propose a CPP model for the release of molecules from an external bio-inspired noise source. A CPP or space-time Poisson process is defined as follows \cite{Hanson}:

\newtheorem{definition}{Definition}
\begin{definition}
	Let $N(t)$ denote a Poisson process characterized by RV $n_e(t,u_t)$ representing the event arrival rate of distribution family $u_t \in U$ at time $t$. Also, let RV $Q(t,w_t)$ denote the time-dependent amplitudes due to a corresponding event at time $t$ from the distribution family $w_t \in W$. Then, a general CPP is defined as the sum of the event amplitudes up to time $t$ and is given by
	\begin{equation}
	\Pi(t)=\int_{0}^{t}Q(\tau,w_{\tau}) n_e(\tau,u_{\tau})d\tau.
	\end{equation}
	
\end{definition}
The release processes observed in secretory bursts of in some endocrine systems \cite{hormone}-\cite{Neuro2} and the release of neurotransmitters in synapses \cite{Synapse2} coincides to CPPs. 
In particular, the authors in \cite{Neuro}-\cite{Neuro2} model the release process of secretory bursts by superimposing the random burst amplitudes on a Poisson process representing the timing of the secretory burst events. The authors' aim is to provide a mathematical model for neurohormone secretion for physiological investigations. Considering Definition 1, this model for secretory bursts constitutes a CPP, regardless of the release amplitude distribution. 
The release of neurotransmitters in synapses which is the fundamental process that drives information transfer between the neurons in the nervous system has been modeled as a doubly stochastic Poisson process in \cite{Synapse2}. A doubly Poisson process is a Poisson process with a random arrival rate. Hence, according to Definition 1, regardless of the distribution of the release amplitudes, doubly Poisson process is a CPP. 

Obviously, the CPP model is  incomprehensive and lacks  to model all biological release processes in body. For instance, Weibull renewal process is a more general model than the CPP, for the release processes in the endocrine systemes \cite{hormone}. More accurately, the Weibull renewal process could consider a range of release processes from the uncorrelated to the fully-correlated time intervals between the release events. However,  the Poisson process is analytically more tractable  because of its specific characteristics, e.g., thinning and memoryless properties. Besides, this specific model could give insightful ideas about more general models.   
 	\definecolor{alizarin}{rgb}{0.82, 0.1, 0.26}
%
%
%
Thereby, in this paper, a CPP model for an external bio-inspired noise source in MC is adopted. Such an external noise source is referred to as a CPNS. 

The resulting CPNS models the randomness of the molecule release process both in time and amplitude. In particular, the molecule release times of the CPNS are  modeled as a Poisson point process with rate $n_e(t,u_t)$ and the number of molecules released at a release  event, $Q(t,w_t)$, is also a RV, which may follow some time-dependent distribution. 

We assume that the CPNS follows a special CPP where the Poisson point process representing the release times has a fixed rate of $\lambda_e$.  This simplifying assumption is justified based on the slow variation of biological processes modulating the release rate of CPNS compared to the transmission time slot duration of DMC systems (usually in the order of seconds).  Moreover, the release amplitude (the number of molecules released at the release time of the CPNS) is assumed to be a Poisson RV with parameter $\lambda_a$. This assumption is confirmed in \cite{Pump} where the authors show that the number of molecules released by ion channels and ion pumps over the cell membrane are Poisson distributed.  For better perception of the logic behind this assumption, consider a chamber including large number of molecules. If each molecule has a small probability to exit during a time interval, the total number of exiting molecules is Poisson distributed RV with mean of the average exiting molecules during the time interval.
In summary, we assume a CPNS with the Poisson point process representing the release times of the fixed rate of $\lambda_e$ where the release event amplitudes are assumed mutually independent and identical Poisson RVs, independent from the  Poisson point process of the release time events.

\definecolor{alizarin}{rgb}{0.82, 0.1, 0.26}

   \color{black}




\subsection{Channel Model}
In the considered DMC system, the transmitter and the CPNS release molecules into the environment. 
The molecules diffuse following a Brownian motion and their movements are assumed to be independent of each other.
Given a molecule $A$ having diffusion coefficient $D$ is released in the described unbounded environment at the origin, $\mathbf{r}=(0,0,0)$ and  at time $t=0$, the probability that the released molecule is observed by a transparent spherical receiver  with volume $V_R=\frac{4}{3}\pi r_R^3$, whose center is at a distance of ${r}$  from the source, can be approximated as\cite{Mahfuz2014}:
\begin{equation}\label{channelimpulse}
\begin{aligned}
p(t)= \frac{V_R}{(4\pi Dt)^{3/2}}\exp\Big(-\frac{{r}^2}{4Dt}\Big)u(t).
\end{aligned}
\end{equation}
It is obvious from \eqref{channelimpulse} that the DMC channel has memory, i.e., a molecule released at the beginning of the current time slot may not be observed at the receiver in the current time slot but may be observed in one of the next time slots.  Theoretically, the DMC channel has infinite memory, since $p(t)$ given in \eqref{channelimpulse} has  an infinite tail. However, from a practical perspective, a finite channel memory can be assumed \cite{memory}. To this end, we define the channel memory as the time it takes until a released molecule arrives at the receiver with a high probability  which is denoted by $\rho$ in this paper, i.e., we have
\begin{equation}
\begin{aligned}
\int_{\tau=0}^{t_m}p(\tau)d\tau=\rho \int_{\tau=0}^{\infty}p(\tau)d\tau,
\end{aligned}
\end{equation}
where $t_m$ denotes the channel memory in seconds. \color{black} Correspondingly,  $k=\lfloor t_m/T\rfloor$  is the channel memory in terms of the number of time slots where $\lfloor x \rfloor$ denotes the largest integer  less than or equal to $x$. 
 The memories of the transmitter-to-receiver channel and the CPNS-to-receiver channel that depend on $d_T$ and $d_C$, respectively, are denoted by $k_T$ and $k_C$, respectively.

\begin{remark}
We have adopted simply an unbounded environment with point source CPNS and transparent receiver for evaluation of DMC in the presence of the CPNS in simulation and numerical results Section. However, our proposed analysis in the rest of paper can be simply generalized for a diffusion channel in which the diffusing molecules are exposed to boundaries of biological entity, the protein receptors over the receiver surface, and/or degradation reactions in the environment. More accurately, in our analysis, the time probability densities of receiving a molecule released from the transmitter and CPNS at the receiver, i.e., $p_T(t)$ and $p_C(t)$ are parameters which are adopted based on the considered system model.   
\end{remark}

\section{Received Signal at the Receiver} \label{section3}

In order to investigate the performance of the considered DMC system, the received signal has to be characterized. In other words, the  distribution of the number of molecules observed at the receiver at  sampling time $t_s$ has to be obtained. The molecules observed at the receiver originate from two independent sources, namely the transmitter and the CPNS. In this section, we derive the distributions of  the numbers of molecules received from the transmitter and the  CPNS, respectively. Specifically, the rare-event property of the Poisson process is employed to obtain a simplified closed-form expression for distribution of the noise received from CPNS. Also, for the special case of CPNS in high-rate regime, the noise received from the CPNS is approximated by a Poisson process whose rate is normal distributed.

\subsection{Signal Received from Transmitter}

Let $B_0 \in \{0,1\}$ and $B_j \in \{0,1\}, j=1,2,\cdots,k_T$, denote the RVs representing the bits transmitted in the current time slot and the $j^{th}$ previous time slot, respectively. Based on the  system model described in Section \ref{system model}, to transmit bit $B_j$, the transmitter releases $X_j$ molecules at the beginning of the $j^{th}$ time slot where $X_j|B_j=b_j \sim \mathrm{Poisson}(b_j N)$. In other words, if $B_j=0$, no molecule is released and if $B_j=1$, the number of released molecules is Poisson distributed with parameter $N$. A molecule released at the beginning of the $j^{th}$, $j=0,1,\cdots,k_T$, time slot is observed at the receiver at sampling time $t_s$ of the current time slot with probability  $p_T(jT+t_s)$, where $p_T(t)$ is given in \eqref{channelimpulse} after substituting $r$ by $d_T$. Based on the thinning property of the Poisson distribution\cite{Thinning}, the number of molecules received at the receiver due to  transmission of $B_0=b_0$ in the current time slot, $Y_T^c$, is Poisson distributed with mean $N b_0 p_T(t_s)$, i.e.,
\begin{equation}
\begin{aligned}
p_{Y_T^c}[k|B_0=b_0]=\mathrm{exp}\bigg(-Nb_0p_T(t_s)\bigg)\frac{\Big(Nb_0p_T(t_s)\Big)^k}{k!}, \hspace{1 cm}  k=0,1,2,\cdots.
\end{aligned}
\end{equation}
Similarly, the number of molecules observed at the receiver  in the current time slot due to  transmission of $B_j=b_j$ in the $j^{th}$ previous time slot is Poisson distributed with parameter $N b_j p_T(jT+t_s)$, i.e., $Y_T^j|B_j=b_j \sim \mathrm{Poisson}\Big(N b_j p_T(jT+t_s)\Big)$. The number of molecules observed at the receiver in the current time slot due to  transmission of all previous bits (interference) equals  $Y_T^I=\sum_{j=1}^{k_T}Y_T^j$. Given the previous transmitted bits $B_j=b_j, j=1,2,\cdots,k_T$, the $Y_T^j$, $j=1,\cdots,k_T$, are independent and  $Y_T^I$ follows a Poisson distribution with mean $N\sum_{j=1}^{k_T}b_j p_T(jT+t_s)$, i.e., we have:
\small
\begin{equation}
\begin{aligned}
p_{Y_T^I}[k|\textbf{B}_{1:k_T}=\textbf{b}_{1:k_T}]=
\exp\bigg(-N\sum_{j=1}^{k_T}b_j p_T(jT+t_s)\bigg)\frac{\Big(N\sum_{j=1}^{k_T}b_j p_T(jT+t_s)\Big)^k}{k!}, \hspace{1 cm}  k=0,1,2,\cdots,
\end{aligned}
\end{equation}
\normalsize
where $\textbf{B}_{1:k_T}=[B_1, B_2, \cdots, B_{k_T}]$, $\textbf{b}_{1:k_T}=[b_1, b_2, \cdots, b_{k_T}]$, and the notation $\textbf{A}_{i:k}$  denotes  vector $[A_i,A_{i+1},\cdots,A_k]$.

\subsection{Noise Received from CPNS}\label{RSANS}


\begin{figure}
	\center
	\includegraphics[width=12 cm,height=4 cm]{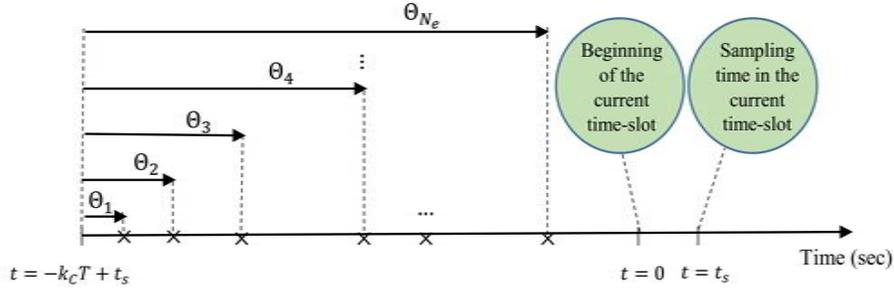}
	\setlength{\abovecaptionskip}{0 cm}
	\caption{A schematic illustration of  Poisson point process over time. The `$\times$' signs over the time axes represent the instances of release events.} 
	\label{fig:events}
\end{figure}

%

The  memory of the channel between the CPNS and the receiver is $k_C$ time slots. Considering the CPNS Poisson point process of release events with rate $\lambda_e$, the number of release events in the $k_C$ previous time slots, $k_C T$, denoted by $N_e(k_CT)$, is a Poisson RV with mean $\lambda_e k_C T$. For  ease of notation, in the rest of the paper, we denote $N_e(k_CT)$ by $N_e$.  The time elapsed since the time instant $-k_CT+t_s$ until the $i^{th}$ release event is denoted by $\Theta_i$; see Fig. \ref{fig:events}. Hence, the $i^{th}$ release event occurs at time $-k_CT+t_s+\Theta_i$. 


Given the release of a molecule by the CPNS at time 0, the probability of observing this molecule at the receiver is $p_C(t)$, given by \eqref{channelimpulse} after substituting  $r$ by $d_C$. The release amplitude of the $i^{th}$ event is Poisson distributed with parameter $\lambda_a$. Therefore, based on the thinning property of the Poisson distribution, the number of molecules observed at the receiver due to the $i^{th}$ release event of the CPNS, $Y_C^i$, is Poisson distributed with mean $\lambda_a p_C(k_CT-\theta_i)$ for $\Theta_i=\theta_i$, i.e.,
\begin{equation}
\begin{aligned}
Y_C^i|\Theta_i=\theta_i \sim \mathrm{Poisson}\Big(\lambda_a p_C(k_CT-\theta_i)\Big).
\end{aligned}
\end{equation}

The total number of molecules observed at the receiver in the current time slot (at time $t=t_s$) due to the molecule release by the CPNS during the $k_C$ previous time slots is equal to $Y_C=\sum_{i=1}^{N_e}Y_C^i$. Assuming $N_e=n$, i.e., the number of release events in the $k_C$ previous time slots is equal to $n$ and $\Theta_{1:n}=\theta_{1:n}$, $Y_C$ follows a Poisson distribution with parameter $\lambda_a \sum_{i=1}^{n}p_C(k_CT-\theta_i)$. Therefore, we can write
\small
\begin{equation}\label{PYIeq}
\begin{aligned}
p_{Y_C}[k]&=\sum_{n=0}^{\infty}\int_{\theta_{1:n}}p_{Y_C}[k|N_e=n, \Theta_{1:n}=\theta_{1:n}]f_{\Theta_{1:n}}(\theta_{1:n}|N_e=n)p_{N_e}[n]d\theta_{1:n} \hspace{2.5 cm}\\
&=\sum_{n=0}^{\infty}\int_{{\theta}_{1:n}}\mathrm{exp}\big(-\lambda_a \sum_{i=1}^{n}p_C(k_CT-{\theta}_i)\big)\frac{\Big(\lambda_a \sum_{i=1}^{n}p_C(k_CT-{\theta}_i)\Big)^k}{k!}f_{{\Theta}_{1:n}}({\theta}_{1:n}|N_e=n)p_{N_e}[n]d{\theta}_{1:n},
\end{aligned}
\end{equation}
\normalsize
where $p_{Y_C}[ \cdot ]$ denotes the distribution of RV $Y_C$, the ${\Theta}_i$s are points of a Poisson process,  $N_e$ is the RV representing the total number of release events occurring during $[-k_C T+t_s,t_s]$, respectively. Also,  $f_{{\Theta}_{1:n}}({\theta}_{1:n}|N_e=n)$ is the joint pdf of release time events $\Theta_{1:n}$ given $n$ events, and $d{\theta}_{1:n}$ stands for $d{\theta}_1 d{\theta}_2 \cdots d{\theta}_n$. Since the points of a Poisson process, $\Theta_{1:n}$,  form a Markov chain\cite{markov}, we have:
\begin{equation}\label{markoveq}
\begin{aligned}
f_{{\Theta}_{1:n}}({\theta}_{1:n}|N_e=n)&=\prod_{i=1}^{n}f_{{\Theta}_i}({\theta}_i|{\theta}_{i-1},N_e=n) \hspace{6.2 cm}\\
&=f_{{\Theta}_1}({\theta}_1|N_e=n)f_{{\Theta}_2}({\theta}_2|{\theta}_{1},N_e=n) \cdots f_{{\Theta}_n}({\theta}_n|{\theta}_{n-1},N_e=n).
\end{aligned}
\end{equation}

The time difference between two events of a Poisson process is exponentially distributed, which leads to:
\begin{equation}
\begin{aligned}
f_{{\Theta}_i}({\theta}_i|{\theta}_{i-1},N_e=n)=\lambda_e e^{-\lambda_e ({\theta}_i-{\theta}_{i-1})}u({\theta}_i-{\theta}_{i-1}),
\end{aligned}
\end{equation}
where $u(\cdot)$ denotes the unit step function. Therefore,  \eqref{markoveq} reduces to
\begin{equation}\label{markoveq2}
\begin{aligned}
f_{{\Theta}_{1:n}}({\theta}_{1:n}|N_e=n)=\lambda_e^n e^{-\lambda_e {\theta}_n} u(\theta_n-\theta_{n-1}).
\end{aligned}
\end{equation}

Substituting $p_{N_e}[n]=e^{-\lambda_ek_CT}{(\lambda_ek_CT)^n}/{n!}$ and applying \eqref{markoveq2} in \eqref{PYIeq}, we have:
\small
\begin{equation} \label{cumbersome}
\begin{aligned}
&p_{Y_C}[k]\\
&=\sum_{n=0}^{\infty}\int_{{\theta}_{1:n}}\exp \bigg(-\lambda_a \sum_{i=1}^{n}p_C(k_CT-{\theta}_i)\bigg)\frac{\Big(\lambda_a \sum_{i=1}^{n}p_C(k_CT-{\theta}_i)\Big)^k}{k!}\lambda_e^n e^{-\lambda_e {\theta}_n}\exp(-\lambda_ek_CT)\frac{(\lambda_ek_CT)^n}{n!}d{\theta}_{1:n}.
\end{aligned}
\end{equation}
\normalsize
Generally, obtaining a closed-form expression for the above integral is cumbersome. 
In the next subsection, the rare-event property \cite{rare event2} of the Poisson process is employed to obtain a simplified closed-form expression for $p_{Y_C}[k]$.  


\subsection{ Rare Event Based  Analysis of the Noise Received from the CPNS}\label{raresubsection}


\color{black}

The rare event property of a Poisson process with parameter $\lambda$ states that the probability of occurrence of an event in a short time interval $\Delta t$ ($\Delta t \ll 1/{\lambda}$) is proportional to the duration of the interval, i.e., $\lambda \Delta t$. Also, for sufficiently small $\Delta t$, the probability of occurrence of more than one event is negligible. As a result, the probability of no event occurring in this interval is equal to $1-\lambda \Delta t$ \cite{Thinning}, \cite{rare event2}.

For a sufficiently short time interval $\tilde{T}$, such that $\lambda_e \tilde{T}\ll 1 $, the rare event property holds. 
Given $\tilde{T}$, the channel memory duration, $k_CT$, can be divided into $\tilde{k}_C={\frac{k_CT}{\tilde{T}}}$ distinct time intervals of length $\tilde{T}$. Let $\tilde{Y}_C^i$ denote the number of molecules received in the current time slot due to the release event of the CPNS  in the $i^{th}$, $i=1,\cdots,\tilde{k}_C$,  previous short time interval, $[-(\tilde{k}_C-i+1)\tilde{T}+t_s,-(\tilde{k}_C-i)\tilde{T}+t_s]$. Therefore, the total number of molecules received from the CPNS in the current time slot  is $Y_C=\sum_{i=1}^{\tilde{k}_C}\tilde{Y}_C^i$. Because of the independence of the release time instants in  distinct intervals for a Poisson process, the $\tilde{Y}_C^i$ are mutually independent, and we have
\normalsize
\begin{equation} \label{conv1}
\begin{aligned}
p_{Y_C}[k]=p_{\tilde{Y}_C^1}[k]\otimes p_{\tilde{Y}_C^2}[k] \otimes  \cdots \otimes p_{\tilde{Y}_C^{\tilde{k}_C}}[k],
\end{aligned}
\end{equation}
where $\otimes$ is the convolution operator and $p_{\tilde{Y}_C^i}[k]$ denotes the distribution of the number of molecules received in the current time slot due to the release event of the CPNS  in the $i^{th}$, $i=1,\cdots,\tilde{k}_C$,  previous short time interval.

In the $i^{th}$  previous short time interval,  the probability of releasing no molecules is  $1-\lambda_e\tilde{T}$. If no molecule is released, which we refer to as event $\mathcal{F}_0^i$, no molecule is observed at the receiver, i.e., $p_{\tilde{Y}_C^i}[k|\mathcal{F}_0^i]=\delta[k]$. Otherwise,  one event occurs in the $i^{th}$ short time interval at time $t_i$, $t_i\in [-(\tilde{k}_C-i+1)\tilde{T}+t_s,-(\tilde{k}_C-i)\tilde{T}+t_s]$, which can be modeled as a uniform RV, since the time interval is short  and includes only one occurrence. Defining $\tilde{\Theta}_i=t_i+\tilde{k}_C \tilde{T}-t_s$ (the time elapsed since time instant $-\tilde{k}_C \tilde{T}+t_s$ until the release event at $t_i$), $\tilde{\Theta}_i$ is a uniform RV in  interval $[(i-1)\tilde{T},i\tilde{T}]$, correspondingly. Thereby, given one event occurrence, $\mathcal{F}_1^i$, and the occurrence time $\tilde{\Theta}_i=\tilde{\theta}_i$, the number of molecules received in the current time slot follows a Poisson distribution with parameter $\mu(\tilde{\Theta}_i)=\lambda_a p_C(k_C T-\tilde{\Theta}_i)$, i.e.,
\begin{equation}\label{release}
\begin{aligned}
p_{\tilde{Y}_C^i}[k|\mathcal{F}_1^i]=\int_{(i-1)\tilde{T}}^{i\tilde{T}} p_{\tilde{Y}_C^i}[k|\mathcal{F}_1^i, \tilde{\theta}_i] \frac{1}{\tilde{T}} d\tilde{\theta}_i=\int_{(i-1)\tilde{T}}^{i\tilde{T}}e^{-\mu(\tilde{\theta}_i)}\frac{\Big(\mu(\tilde{\theta}_i)\Big)^k}{k!}\frac{1}{\tilde{T}} d\tilde{\theta}_i,  \;\; \; \; \; \; i=1,2,...,\tilde{k}_C.
\end{aligned}
\end{equation}
\normalsize
Therefore, we obtain
\begin{equation}\label{16}
\begin{aligned}
p_{\tilde{Y}_C^i}[k]=p_{\tilde{Y}_C^i}[k|\mathcal{F}_0^i]p(\mathcal{F}_0^i)&+p_{\tilde{Y}_C^i}[k|\mathcal{F}_1^i]p(\mathcal{F}_1^i)\\
&=(1-\lambda_e \tilde{T})\delta[k]+(\lambda_e \tilde{T})p_{\tilde{Y}_C^i}[k|\mathcal{F}_1^i], \;\; \; \; \; \; i=1,2, \cdots,\tilde{k}_C,
\end{aligned}
\end{equation}
where $p_{\tilde{Y}_C^i}[k|\mathcal{F}_1^i]$ is given by \eqref{release}. To simplify the notation, we employ $f_i[k]=p_{\tilde{Y}_C^i}[k|\mathcal{F}_1^i]$ in the rest of the paper. Eq. \eqref{conv1} can be simplified as follows, see  Appendix \ref{App-A},
\begin{equation} \label{conv}
\begin{aligned}
p_{Y_C}[k]=\sum_{i=0}^{\tilde{k}_C}(1-\lambda_e \tilde{T})^{\tilde{k}_C-i}(\lambda_e \tilde{T})^i \sum_{h=1}^{\mathfrak{K}_i}  \delta[k] \otimes {f}_{\alpha_1^h}[k] \otimes \cdots \otimes {f}_{\alpha_i^h}[k],
\end{aligned}
\end{equation}
where  $\mathfrak{K}_i \overset{\Delta}{=}\left(\begin{array}{c}\tilde{k}_C\\ i\end{array}\right)$ is the number of $i$-element subsets of  set $\{1,2, \cdots,\tilde{k}_C\}$, and the elements of the $h^{th}$ $i$-element subset are denoted by $ \alpha_1^h, \cdots, \alpha_{i}^h$. The distribution in \eqref{conv} is complicated and does not have a closed form expression when the exact distribution of $p_{\tilde{Y}_C^i}[k]$ given in \eqref{release} is employed which is referred to as rare-event exact distribution for the noise received from the CPNS. To obtain a closed form expression, we approximate $p_{\tilde{Y}_C^i}[k|\mathcal{F}_1^i]$ as a Poisson distribution with a fixed mean. Since $\tilde{\Theta}_i$ is a uniform RV in  interval $[(i-1)\tilde{T},i\tilde{T}]$, by adopting sufficiently small $\tilde{T}$ ($\tilde{T}\ll k_C T$), $k_CT-(i-1)\tilde{T}$, very closely approximates $k_C T-\tilde{\Theta}_i$. Therefore, we can approximate the mean $\mu(\tilde{\Theta}_i)$ as follows
\begin{equation}\label{app2}
\begin{aligned}
\mu(\tilde{\Theta}_i)=\lambda_a p_C(k_C T-\tilde{\Theta}_i) \simeq  \lambda_a p_C(k_CT-(i-1)\tilde{T})\overset{\Delta}{=} \mu_i, \;\; \; \; \; \; i=1,2, \cdots,\tilde{k}_C,
\end{aligned}
\end{equation}
which leads to ${f}_i[k] = p_{\tilde{Y}_C^i}[k|\mathcal{F}_1^i\color{black}] \cong e^{-\mu_i}\frac{(\mu_i)^k}{k!}$. Thereby, the convolution term in \eqref{conv} is reduced to
\begin{equation}
\begin{aligned}
{f}_{\alpha_0^h}[k] \otimes {f}_{\alpha_1^h}[k] \otimes \cdots \otimes {f}_{\alpha_i^h}[k] \cong \exp \bigg(- \sum_{l=0}^{i}\mu_{\alpha_l^h}\bigg)  \frac{( \sum_{l=0}^{i}\mu_{\alpha_l^h})^k}{k!},
\end{aligned}
\end{equation}
and hence
\begin{equation}\label{PM noise}
\begin{aligned}
p_{Y_C}[k]=\sum_{i=0}^{\tilde{k}_C}(1-\lambda_e \tilde{T})^{\tilde{k}_C-i}(\lambda_e \tilde{T})^i \sum_{h=1}^{\mathfrak{K}_i} \exp \bigg(- \sum_{l=0}^{i}\mu_{\alpha_l^h}\bigg) \frac{( \sum_{l=0}^{i}\mu_{\alpha_l^h})^k}{k!},
\end{aligned}
\end{equation}
which is a Poisson mixture distribution\cite{Poisson Mixture2}, i.e., $p_{Y_C}[k]$ is a summation of weighted Poisson distributions where the sum of the weights is equal to 1, since
\begin{equation}
\begin{aligned}
\sum_{i=0}^{\tilde{k}_C}(1-\lambda_e \tilde{T})^{\tilde{k}_C-i}(\lambda_e \tilde{T})^i\mathfrak{K}_i=1.
\end{aligned}
\end{equation}
We refer to the approximation in \eqref{PM noise} as rare-event approximate distribution for noise received from the CPNS. Our simulation results demonstrate that the rare-event approximate analysis very closely approaches the rare-event exact analysis in \eqref{conv} for small values of $\tilde{T}$.    

\color{black}


%


\color{black}

\subsection{ Noise Received from the CPNS in the High-rate Regime}\label{III.D}

From  Subsection \ref{RSANS}, the number of molecules received from the CPNS at the current time slot follows a Poisson distribution with random rate  $M=\lambda_a \sum_{i=1}^{N_e}p_C(k_CT-{\Theta}_i)$ where $N_e$ is a RV denoting the number of release events during the $k_C$ previous time slots and $\Theta_i$ denotes a RV representing the release time of the $i^{th}$ event with respect to the beginning of the $k_C^{th}$ previous time slot. Defining stochastic process $\textbf{M}(t)=\sum_{i=1}^{N_e}\lambda_a p_C(t-{\Theta}_i)$,  we can write $M=\textbf{M}(k_CT)$. The $\textbf{M}(t)$ can be interpreted as a shot-noise process passed over a linear time invariant (LTI) system with impulse response $\lambda_a p_C(t)$ \cite{NEW2}-\cite{NEW3}, i.e.,
\begin{equation} \label{fpg}
\textbf{M}(t)=\sum_{i=-\infty}^{+\infty}\delta(t-\Theta_i)\otimes \lambda_a p_C(t),
\end{equation}
where $p_C(t)$ is given in \eqref{channelimpulse} and ${\Theta}_i$s are points of a Poisson process with rate $\lambda_e$\footnote{Since we have $p_C(t)=0$ for $t<0$ and $t>k_CT$, inclusion of $\Theta_i<0$ or $\Theta_i>k_CT$ in the summation of $\textbf{M}(t)=\sum_{i=1}^{N_e}\lambda_a p_C(t-{\Theta}_i)$ is allowed and equivalently we can write $\textbf{M}(t)=\sum_{i=-\infty}^{+\infty} \lambda_a p_C(t-\Theta_i)$}. The  cumulants  of a shot-noise process passing from an LTI system with impulse response $h(t)$ are time invariant which are given by \cite{NEW3}:
\begin{equation} \label{kkn}
\begin{aligned}
k_n=\lambda_e \int_{-\infty}^{+\infty}h^n(\tau) d\tau,\;\;\;\; \forall\; n\geq1.
\end{aligned}
\end{equation}

Thereby, this process is a first order strict sense stationary (SSS) process whose first order distribution function is time independent. The authors in \cite{NEW2}  show that for high values of $\lambda_e$ ($\lambda_e\to \infty$), the first order distribution of this process approaches a Gaussian distribution with mean $k_1$ and variance $k_2$ given in \eqref{kkn}. Considering $h(t)=\lambda_ap_C(t)$ where $p_C(t)$ is given in \eqref{channelimpulse}, the cummulants of $\textbf{M}(t)$ are obtained as follows
\begin{equation} 
\begin{aligned} \label{k_n2}
k_n=\frac{1}{n}G_1G_2^n \Gamma\big(\frac{3n}{2}-1,\frac{nd_C^2}{4Dk_CT}\big),\;\;\forall n\geq 1,
\end{aligned}
\end{equation}
where $G_1=\frac{\lambda_e d^2}{4D}$, $G_2=\frac{\lambda_a V_R}{\pi^{3/2}d_C^3}$, and $\Gamma(s,x)=\int_x^{\infty}t^{s-1}e^{-t}dt$ denotes the upper incomplete Gamma function. Therefore, for CPNS in a high-rate regime (large values of $\lambda_e$), ($\lambda_e\to \infty$), $M$ follows a Gaussian distribution with mean $k_1$ and variance $k_2$ given in \eqref{k_n2}, i.e.,
\begin{equation} 
\begin{aligned} \label{gauus}
f_M(m)= (2\pi k_2)^{1/2}\mathrm{exp}\big(-\frac{(m-k_1)^2}{2k_2}\big),
\end{aligned}
\end{equation}
As a result, the number of molecules received from the CPNS in high rate regime, $Y_C$, follows a Poisson distribution with parameter $M\sim \mathcal{N}(k_1,k_2)$, and we can write:
\begin{equation} 
\begin{aligned} \label{abc}
{p}_{Y_C}[k]=\int_{0}^{+\infty} {p}_{Y_C}[k|m]f_M(m)dm.
\end{aligned}
\end{equation}
In Appendix B, we obtain the following closed form expression for $p_{Y_C}[k]$
\begin{equation}  
\begin{aligned} \label{distinf}
{p}_{Y_C}[k]={l}{(k_1,k_2)} k_2^{(k+1)/2}D_{-k-1}\bigg(\sqrt{k_2}(1-\frac{k_1}{k_2})\bigg),
\end{aligned}	
\end{equation}
where 
\begin{equation}  
\begin{aligned}
l{(k_1,k_2)}\overset{\Delta}{=} (2\pi k_2)^{-1/2}\exp(-k_1^2/{2k_2}+k_2(1-\frac{k_1}{k_2})^2/4),
\end{aligned}	
\end{equation}
and $D_{\nu}(z)$ is the parabolic cylinder function which is defined as follows
\begin{equation}  \label{parabolic} 
\begin{aligned}
D_{\nu}(z)\overset{\Delta}{=} 2^{\nu/2}e^{-z^2/4}\bigg[\frac{\sqrt{\pi}}{\Gamma(\frac{1-\nu}{2})}   {_{1}F_1}(-\frac{\nu}{2},\frac{1}{2};\frac{z^2}{2})-\frac{\sqrt{2\pi}z}{\Gamma(-\frac{\nu}{2})}\ _{1}F_1(\frac{1-\nu}{2},\frac{3}{2};\frac{z^2}{2})\bigg],
\end{aligned}	
\end{equation}
in which $_{1}F_1(\alpha,\gamma;z)$ denotes the confluent hypergeometric function as follows:
\begin{equation}  
\begin{aligned} 
_{1}F_1(\alpha,\gamma;z)=\sum_{n=0}^{+\infty}\frac{\Gamma(\alpha+n)\Gamma(\gamma)}{\Gamma(\alpha)\Gamma(\gamma+n)}\frac{z^n}{n!}.
\end{aligned}	
\end{equation}
where $\Gamma(.)$ denotes the Gamma function.

\section{Error Probability Analysis}\label{section4}



In Section \ref{section3}, the  signal received from the transmitter, $Y_T$, and  the noise received from the CPNS, $Y_C$, were analyzed. In this section, we analyze the performance of a point-to-point DMC link in the presence of a CPNS in terms of the BER. The total received signal at the receiver is given by
\begin{equation}
\begin{aligned}
Y=Y_T+Y_C=Y_T^c+Y_T^I+Y_C,
\end{aligned}
\end{equation}
where $Y_T^c$ is the  signal received in the current time slot due to the current transmission and $Y_T^I$ is the interference received in the current time slot originating from transmissions in previous time slots.    
$Y_T^c$ and $Y_T^I$ are independent Poisson-distributed RVs with parameters $Nb_0p_T(t_s)$ and $N\sum_{j=1}^{k_T}b_jp_T(jT+t_s)$, given $B_0=b_0$ and $\textbf{B}_{1:k_T}=\textbf{b}_{1:k_T}$, respectively.
 Thereby, $Y_T=Y_T^c+Y_T^I$ follows a Poisson distribution with parameter $\sum_{j=0}^{k_T}NB_jp_T(jT+t_s)$, i.e.,
\begin{equation}
\begin{aligned}
p_{Y_T}[k|\textbf{B}_{0:k_T}=\textbf{b}_{0:k_T}]
=\exp\Bigg(-\sum_{j=0}^{k_T}Nb_jp_T(jT+t_s)\Bigg)\frac{\Big(\sum_{j=0}^{k_T}Nb_jp_T(jT+t_s)\Big)^k}{k!}.
\end{aligned}
\end{equation} 
\normalsize

The noise received from the CPNS, $Y_C$,  is a Poisson mixture given by \eqref{PM noise}. Therefore, conditioned on current and previous transmitted bits, $\textbf{B}_{0:k_T}=\textbf{b}_{0:k_T}$,  the total number of molecules observed at the receiver also follows a Poisson mixture distribution as follows:
\begin{equation} \label{received signal}
\begin{aligned}
p_{Y}[k|\textbf{B}_{0:k_T}=\textbf{b}_{0:k_T}]=\sum_{i=0}^{\tilde{k}_C}(1-\lambda_e \tilde{T})^{\tilde{k}_C-i}(\lambda_e\tilde{T})^i \sum_{h=1}^{\mathfrak{K}_i} e^{-\upsilon_i^{h}(\textbf{b}_{0:k_T})} \frac{\Big(\upsilon_i^{h}(\textbf{b}_{0:k_T})\Big)^k}{k!},
\end{aligned}
\end{equation}
where $\upsilon_i^{h}(\textbf{b}_{0:k_T})=\sum_{j=0}^{k_T}Nb_jp_T(jT+t_s)+\sum_{l=0}^{i}\mu_{\alpha_l^h}$. 
Assuiming equiprobable input bits and receiving $Y=y$ molecules in the current time slot, a symbol-by-symbol maximum likelihood (ML) detector which has no information about the previously transmitted bits is given by  \cite{a1}-\cite{salehi}:
\begin{equation} \label{ML22}
\begin{aligned}
\hat{B}_0=\underset{b_0 \in \{0,1\}}{\operatorname{argmax}} \hspace{.25 cm} p_Y[y|{B}_0={b}_0],
\end{aligned}
\end{equation}
where
\begin{equation}\label{STDpe6} 
\begin{aligned}
p_Y[y|{B}_0={b}_0]= (\frac{1}{2})^{k_T}\sum_{\textbf{b}_{1:k_T}}p_Y[y|\textbf{B}_{0:k_T}=\textbf{b}_{0:k_T}],
\end{aligned}
\end{equation}
and $p_Y[y|\textbf{B}_{0:k_T}=\textbf{b}_{0:k_T}]$ is given in \eqref{received signal}.
Generally, the optimal ML detector \eqref{ML22} is a MTD which is characterized by $m$ threshold values,  $\zeta_1,\zeta_2,\cdots,\zeta_m$ partitioning feasible observation space ($y\in \mathbb{R}_{+}$) and the decisions on the transmitted bit based on the observed $y$ in all disjoint partitions determined by the threshold values. For nanomachines, which have limited resources, STDs ($m=1$) are desirable.  A  STD, denoted by $\Phi(y)$, is characterized as follows:
\begin{equation} 
\begin{aligned} \label{det}
\Phi(y)=\begin{cases}0 & y<\zeta\\1 & y\geq\zeta\end{cases},
\end{aligned}
\end{equation}
where $y$ is the observation and $\zeta$ is the decision threshold. 
Given a STD with threshold value  $\zeta$, the BER of the system is obtained as follows:
\begin{equation}\label{STDpe1} 
\begin{aligned}
\mathrm{P}_e=(\frac{1}{2})^{k_T} \sum_{\textbf{b}_{1:k_T}}\mathrm{Pr}(E|\textbf{B}_{1:k_T}=\textbf{b}_{1:k_T}),
\end{aligned}
\end{equation} 
in which $E$ is the error event ($\hat{B}_0\neq B_0$) and BER conditioned to the previous transmitted bits, $\mathrm{Pr}(E|\textbf{B}_{1:k_T}=\textbf{b}_{1:k_T})$, is given by:
\begin{equation}\label{STDpe3} 
\begin{aligned}
\mathrm{Pr}(E|\textbf{B}_{1:k_T}=\textbf{b}_{1:k_T})&=\frac{1}{2}\mathrm{Pr}(\hat{B}_0=1|B_0=0, \textbf{B}_{1:k_T}=\textbf{b}_{1:k_T})+ \frac{1}{2}\mathrm{Pr}(\hat{B}_0=0|B_0=1, \textbf{B}_{1:k_T}=\textbf{b}_{1:k_T})\\
&=\frac{1}{2}\mathrm{Pr}(y\geq \zeta|B_0=0, \textbf{B}_{1:k_T}=\textbf{b}_{1:k_T})+\frac{1}{2}\mathrm{Pr}(y< \zeta|B_0=1, \textbf{B}_{1:k_T}=\textbf{b}_{1:k_T})\\
&=\frac{1}{2}\Big(1-\mathcal{F}_Y(\zeta|B_0=0, \textbf{B}_{1:k_T}=\textbf{b}_{1:k_T})+\mathcal{F}_Y(\zeta|B_0=1, \textbf{B}_{1:k_T}=\textbf{b}_{1:k_T})\Big),
\end{aligned}
\end{equation} 
where $\mathcal{F}_Y(\cdot)$ denotes the cumulative distribution function (CDF) of RV $Y$.

Considering the distribution of received signal given by \eqref{received signal}, the BER terms $\mathrm{Pr}(\hat{B}_0=1|B_0=0,\textbf{B}_{1:k_T}=\textbf{b}_{1:k_T})$ and $\mathrm{Pr}(\hat{B}_0=0|B_0=1,\textbf{B}_{1:k_T}=\textbf{b}_{1:k_T})$  in \eqref{STDpe3} are calculated as follows:
\begin{equation}\label{P_e1}
\begin{aligned}
\mathrm{Pr}(\hat{B}_0=1|B_0&=0, \textbf{B}_{1:k_T}=\textbf{b}_{1:k_T})=\mathrm{Pr}(y>\zeta| B_0=0,\textbf{B}_{1:k_T}=\textbf{b}_{1:k_T}) \\
&= \sum_{i=0}^{\tilde{k}_C}(1-\lambda_e \tilde{T})^{\tilde{k}_C-i}(\lambda_e\tilde{T})^i \sum_{h=1}^{\mathfrak{K}_i} \Bigg(1-\frac{\Gamma\Big(\zeta,\upsilon_i^{h}(b_0=0,\textbf{b}_{1:k_T})\Big)}{\Gamma(\zeta)}\Bigg),\;\;\;\;\;\;\;\;
\end{aligned}
\end{equation}
\begin{equation}\label{P_e2}
\begin{aligned}
\mathrm{Pr}(\hat{B}_0=0|B_0&=1, \textbf{B}_{1:k_T}=\textbf{b}_{1:k_T})=\mathrm{Pr}(y \leq \zeta| B_0=1, \textbf{B}_{1:k_T}=\textbf{b}_{1:k_T}) \\
&=\sum_{i=0}^{\tilde{k}_C}(1-\lambda_e \tilde{T})^{\tilde{k}_C-i}(\lambda_e \tilde{T})^i \sum_{h=1}^{\mathfrak{K}_i}\frac{\Gamma\Big(\zeta,\upsilon_i^{h}(b_0=1,\textbf{b}_{1:k_T})\Big)}{\Gamma(\zeta)}, \hspace{2.8 cm}
\end{aligned}
\end{equation}
\normalsize
where $\Gamma (\delta,\sigma)$ is the incomplete Gamma function given by $\Gamma(\delta,\sigma)=\int_{\sigma}^{\infty}e^{-t}t^{\delta-1}dt$ and $\Gamma(\delta,\sigma)/\Gamma(\delta)$ denotes the CDF of the Poisson distribution with parameter $\sigma$.

\subsection{On the Optimality of Single-Threshold Detector}
In this subsection, we first prove that STD is optimal for CPNS in high-rate regime and then discuss on optimality of STD in general case. 


\newtheorem{theorem}{Theorem}
\begin{theorem}
For the CPNS in the high-rate regime (large values of $\lambda_e$), the optimal ML detector \eqref{ML22} is a single-threshold detector.
\end{theorem}
\begin{proof}
	The detector is supposed to detect the transmitted bits 1 or 0 which is equivalent to the presence or absence of signal $\mathrm{Poisson}(Np_T(t_s))$ which is embedded in the noise $Y_T^I+Y_C$. $Y_T^I$ is Poisson distributed with mean $\sum_{j=1}^{k_T}NB_jp_T(jT+t_s)$ independent from $Y_C$. For CPNS in high-rate regime, we showed in Subsection \ref{III.D} that $Y_C$, follows a Poisson distribution with rate  $M\sim \mathcal{N}(k_1,k_2)$. Therefore, given the previously transmitted bits, $B_{1:k_T}=b_{1:k_T}$, the additive noise is distributed as $Y_{IC}=Y_T^I+Y_C \sim \mathrm{Poisson}(M')$ where $M'\sim \mathcal{N}(k_1',k_2)$ and $k_1'=k_1+\sum_{j=1}^{k_T}Nb_jp_T(jT+t_s)$.
	Therefore, we have 
	\begin{equation} 
	\begin{aligned}
	{p}_{Y_{IC}}[k]=(\frac{1}{2})^{k_T} \sum_{\textbf{b}_{1:k_T}} \int_{0}^{+\infty} {p}_{Y_{IC}}[k|m,B_{1:k_T}=b_{1:k_T}]f_{M'}(m|B_{1:k_T}=b_{1:k_T})dm,
	\end{aligned}
	\end{equation}
where ${p}_{Y_{IC}}[k|m,B_{1:k_T}=b_{1:k_T}]={p}_{Y_{IC}}[k|m]=\exp(-m)\frac{ m^k}{k!}$ and $f_{M'}(m|B_{1:k_T}=b_{1:k_T})=(2\pi k_2)^{-1/2}\mathrm{exp}\big(-\frac{(m-k_1')^2}{2k_2}\big)$. It is easy to see that ${p}_{Y_{IC}}[k|m]$ is a log-concave distribution in terms of $(k,m)$, $f_{M'}(m|B_{1:k_T}=b_{1:k_T})$ is log-concave distribution in terms of $m$ and $b_{1:k_T}$, and also uniform distribution of $p_{B_{1:k_T}}[b_{1:k_T}]=(\frac{1}{2})^{k_T}$ is log-concave distribution. Since pointwise multiplication of log-concave functions is log-concave \cite{boyd}, joint distribution of $(Y_{IC},M',B_{1:k_T})$, i.e., $(\frac{1}{2})^{k_T} {p}_{Y_{IC}}[k|m,b_{1:k_T}]f_{M'}(m,b_{1:k_T}) $, is log-concave. \color{black} Moreover, the marginal distributions of log-concave joint distribution is a log-concave \cite{boyd}. Therefore, ${p}_{Y_{IC}}[k]$ is a log-concave distribution.
	On the other hand, the optimal ML detector of the presence of signal embedded in the additive log-concave noise is a single-threshold detector \cite{Pruc}.   
\end{proof}

Now, we derive the BER of  the considered  DMC system in the presence of CPNS in high-rate regime. Obviously, given $\textbf{B}_{0:k_T}=\textbf{b}_{0:k_T}$, $Y=Y_T^c+Y_T^I+Y_C$ is a Poisson RV whose mean is $M''=M+\sum_{j=0}^{k_T}Nb_jp_T(jT+t_s)$. Since $M\sim \mathcal{N}(k_1,k_2)$, we have $M''\sim \mathcal{N}(k_1'',k_2)$ in which $k_1''=k_1+\sum_{j=0}^{k_T}Nb_jp_T(jT+t_s)$. Considering \eqref{distinf}, we can write 
\begin{equation}\label{distinf2}  
\begin{aligned} 
{p}_{Y}[k|\textbf{B}_{0:k_T}=\textbf{b}_{0:k_T}]={l}{(k_1'',k_2)} k_2^{(k+1)/2}D_{-k-1}\bigg(\sqrt{k_2}(1-\frac{k_1''}{k_2})\bigg).
\end{aligned}	
\end{equation}
Therefore, given a STD with threshold $\zeta$, error probability is given by
\begin{equation}\label{STDpe2} 
\begin{aligned}
\mathrm{P}_e=(\frac{1}{2})^{k_T+1}\sum_{\textbf{b}_{1:k_T}}\Big(1-\mathcal{F}_Y(\zeta|B_0=0,\textbf{B}_{1:k_T}=\textbf{b}_{1:k_T})+\mathcal{F}_Y(\zeta|B_0=1,\textbf{B}_{1:k_T}=\textbf{b}_{1:k_T})\Big),
\end{aligned}
\end{equation} 
where $\mathcal{F}_Y(\cdot)$ denotes the CDF of RV $Y$ whose pdf is given by \eqref{distinf2}.  To obtain the optimal threshold value, one should minimize BER  in \eqref{STDpe2} in terms of $\zeta$. The following simple lemma concludes that when STD is an optimal ML detector, the corresponding error probability  is quasiconvex. 
Therefore, numerical iterative algorithms such as bisection method can be employed to obtain the optimal threshold value \cite{azmi2}-\cite{boyd}.

 
\color{black}

\begin{lem} \label{theorem2} 
The optimal ML detector in \eqref{ML22} is single-threshold with optimal threshold $\zeta_o$, if and only if the BER in \eqref{STDpe2} is a quasiconvex function of the threshold with global minimum at $\zeta_o$.
\end{lem}
\begin{proof} The proof is provided in Appendix C.
\end{proof}


 \newtheorem{corollary}{Corollary} 
\begin{corollary}
Obviously, when the BER is not a quasiconvex function of $\zeta$, it has multiple local minimum and maximum points at $\zeta_1,\zeta_2,\cdots,\zeta_m$ characterizing optimal MTD.
\end{corollary}
Optimality of STD in special case of large values of $\lambda_e$ is borrowed from the log-concavity of normal distribution of $M$ that results the log-concavity of $Y_C$. But, in the general case of CPNS, our results indicate that distributions of $M$ and then $Y_C\sim \mathrm{Poisson}(M)$ may not be log-concave and may even be multimodal in some conditions. Thereby, the distributions of received signals given $b_0=0,$ and 1 in \eqref{STDpe6} may have multiple intersection points and correspondingly leading to optimality of a MTD (and not a STD). 
Fig. \ref{h3} (Left) depicts the logarithm of distribution of noise $Y_C$ obtained based on simulation, where $N=2 \times 10^5 $, $T=0.2 $, $d_T=10 \; \mu m$, $d_C=5.5 \; \mu m$, $\lambda_a=5 \times 10^5$, $\lambda_e=15$, $r_R=2.2 \; \mu m$. It is obvious that it is not a concave curve and then   the logarithm of distribution of noise is not a log-concave distribution. Correspondingly, Fig. \ref{h3} (Right) shows the distributions of the number of received molecules given the transmission of bits 1 and 0 in the current time slot, i.e., $p_Y[k|B_0=1]$ and $p_Y[k|B_0=0]$ given by \eqref{STDpe6}, respectively. 
It is observed that these two distributions are bimodal and have 3 intersection points and then  the  optimal ML detector has 4 decision making regions which results in a MTD. 
\begin{remark}
	Based on our vast numerical and simulation results, this  phenomenon  rarely occurs for considered MC system in the presence of CPNS and is only theoretically of interest. Note that even for such rare scenarios, the difference between BER of the suboptimal STD and the  optimal MTD would be negligible, as it is deduced from Fig. \ref{h3}.
\end{remark}

 	\definecolor{alizarin}{rgb}{0.82, 0.1, 0.26}
\section{Numerical and Simulation Results}\label{section5}

In this section, we evaluate the performance of the point-to-point DMC system in the presence of a CPNS employing a simple OOK modulation.  We have employed the PBS introduced in \cite{PBS} for analysis. To perform the PBS, the time is divided into small time steps $\Delta t$ s. 
The molecule locations are known and the molecules move independently in the 3-
dimensional space in the PBS. In each dimension, the displacement of a molecule in
$\Delta t$ s is modeled as Gaussian distribution with zero mean and variance $2D\Delta t$.  The DMC system parameters adopted for the analytical and simulation are given in Table \ref{table12}.

Fig. \ref{fig:comp} shows the BER of the DMC system in the presence of a CPNS as a function of  the time interval $\tilde{T}$ used for the rare-event analysis  obtained based on (i) the rare-event exact analysis given in \eqref{conv}, (ii) the rare-event approximate analysis given in \eqref{PM noise} and (iii) PBS.  The distance between the CPNS and the receiver and the corresponding channel memory are  $d_C=8 \;\mu m$ and   $k_C=10$, respectively. The event amplitude of the CPNS, which we refer to as the CPNS amplitude, is set to $\lambda_a=10^5$. The BER curves are plotted for two different CPNS rates,  $\lambda_e=2$ \text{and}\;$10$. We observe that both approximation and exact analytical results approach the PBS results for sufficiently short time intervals $\tilde{T}$ for which rare-event property holds ($\lambda_e \tilde{T}\ll 1$). \color{black} 
 Also, it is observed that the rare-event analysis deviates from the corresponding PBS result for higher $\tilde{T}$ values, since  the condition ${\lambda \tilde{T}} \ll 1$ is not well satisfied leading to the rare event property of Poisson distribution does not hold.




Fig. \ref{f2} depicts the BERs of the DMC system in the presence of a CPNS obtained from the rare-event approximate analysis and the PBS. As observed, the PBS results confirm the proposed analysis. \color{black}
Also, this figure compares the BERs of the DMC system for a CPNS (the rare-event approximate analysis) and a homogeneous Poisson noise \cite{Arjmandi2013}. For a higher accuracy of rare-event analysis for the CPNS, we use  very small value of $\tilde{T}=0.002$.  
To have a fair comparison, the average mean of the homogeneous Poisson noise received in the current time slot, denoted by $\lambda_0$, is set equal to the average number of molecules received from the CPNS in the current time slot, i.e., 
\begin{equation}
\begin{aligned}
\hat{\lambda}_0=\sum_{n=0}^{\infty} \int_{{{\theta}_{1:n}}}\lambda_a \sum_{i=1}^n p_C(k_CT-\theta_i)f_{\theta}(\theta_{1:n})p_{N_e}[n]d\theta_{1:n}.
\end{aligned}
\end{equation}

The BER is depicted versus the threshold value $\zeta$ for  $\lambda_e= 0.5$ and $\lambda_e=5$. We  use  $d_C=8 \; \mu m$.
Fig. \ref{f2} reveals that assuming a homogeneous Poisson noise at the receiver leads to an overly optimistic performance prediction when the noise source is actually a CPNS. Therefore, homogeneous Poisson noise models are not capable of modeling CPNSs. Furthermore, a CPNS with a lower rate ($\lambda_e=0.5$), results in a higher performance than a CPNS with a higher rate ($\lambda_e=5$), as expected. Also, it is observed that the simulation results confirm the provided analysis.

Fig. \ref{h1}  compares the distribution of $M=\lambda_a \sum_{i=1}^{N_e}p_C(k_CT-{\Theta}_i)$ obtained based on the simulation with normal distribution approximation $\mathcal{N}(k_1,k_2)$ given in \eqref{gauus}, for different values of $\lambda_e$. It is observed that the distribution of $M$ approaches the normal distribution for high values of $\lambda_e$ ($\lambda_e>100$), confirming our analysis provided in Section \ref{III.D}.   
Correspondingly, Fig. \ref{h2} depicts the BER of the DMC system in the presence of a CPNS versus $\lambda_e$ obtained based on Monte Carlo simulation\footnote{ A Monte-Carlo simulation has been employed,   since  applying  the PBS takes very long time for large values of $\lambda_e$.} and  analysis given in \eqref{STDpe2} which employs normal approximation of $\mathcal{N}(k_1,k_2)$, for different values of $T=0.02, 0.1$. It is observed that the simulation results coincide our analysis for high values of $\lambda_e$ ($\lambda_e>100$).  
 As Fig. 6 confirms, Gaussian approximation is not accurate for smaller $\lambda_e$ values  which leads to the increasing gap between the BERs obtained from the high-rate analysis and simulation results.  Moreover, we observe that BER increases by increasing $\lambda_e$.

%

Fig. \ref{f4} shows the BER of the DMC system in the presence of a CPNS in low and moderate rate regime ($\lambda_e=0.5, 5, 10$) versus threshold $\zeta$, for various parameters including $\lambda_a$, $N$, and $d_C$. The rare-event approximation  analysis was used. The distance between the transmitter and the receiver is fixed to $d_T=4 \; \mu m$. For all five considered scenarios, the BER is a quasiconvex function of $\zeta$, and considering Lemma 1, the optimal detector is a simple STD.

Fig. \ref{f8} depicts the BER of the DMC system versus CPNS amplitude rate ($\lambda_a$) for different values of $d_C=6, 12, 25, 50,  100\; \mu m$. As expected, by increasing the distance of the CPNS from the receiver, $d_C$, the BER decreases which results in a better performance. For small values of $\lambda_a$ such as $\lambda_a \leq 10^3$, the performance is approximately the same for all distances as the impact of the CPNS  becomes negligible for extremely low $\lambda_a$s. In these cases, the randomness of the diffusion channel between the transmitter and receiver is the dominant effect on  the performance. On the other hand, when we have higher $\lambda_a$ values or smaller $d_C$s, sensitivity of BER to both values of $\lambda_a$ and $d_C$ increases.

\section{Conclusion}\label{section6}
In this paper, impact of the presence of biological external noise source for DMC system was investigated. The release of molecules by a biological noise source was particularly modeled as a CPP, inspired by the release processes of some biological entities. \color{black} A point-to-point DMC link in the presence of a CPNS was considered. The distribution of the number of molecules received from CPNS was analyzed based on the rare event distribution and shown to be a Poisson mixture distribution.   
Assuming a simple on-off keying modulation, symbol-by-symbol ML detector was formulated and BER was analyzed in closed-form expressions. For special case of CPNS in high-rate regime, the noise received from the CPNS is approximated by a Poisson process whose rate is normally  distributed. It was proved that the optimal ML detector is a simple STD, for CPNS in high rate regime.
However, our results revealed that STD may not be the optimal ML detector, in the general case of CPNS. Moreover, based on our results, the  presented model and our analysis for the DMC system performance in the presence of a CPNS is necessary and simply adopting conventional homogeneous Poisson noise model may lead to overly optimistic performance predictions. This new type of noise source introduces diffusion channel different from conventional MC channels that should be investigated from various perspectives.
 Analyzing the capacity of this new introduced diffusion channel in the presence of a CPNS, the extension of single CPNS source to the multiple CPNS sources, employing more complicated modulation schemes, and considering more realistic assumptions for the biological environment are left for future works. Moreover, analyzing the impact of biological noise sources with more complicated release processes on the DMC system performance remains open.


\begin{table}
	\footnotesize
	 \caption{DMC system parameters used in simulations} 
	\begin{center}
		\begin{tabular}{|c|c|c|}
			\hline
			\bfseries{Parameter} & \bfseries{Variable} & \bfseries{Value}   \\
			\hline \hline
			Diffusion coefficient & $D$ & $1.14 \times 10^{-9}\;m^2/sec$\\
			\hline
			Time-slot duration & $T$ & $0.5, 0.2, 0.1,$ and $0.02$ \; $sec$ \\
			\hline
			Number of transmitted molecules for bit '1'& $N$ & $0.5 \times 10^5$\\
			\hline
			Distance between transmitter and receiver & $d_T$ & $4 \; \mu m$\\
			\hline
			Distance between CPNS and receiver &  $d_C$  &$6, 8, 12, 25,50$, and $ 100 \; \mu m$ \\
			\hline    
			memory of transmitter-to-receiver channel for $\rho=0.95$   & $k_T$ & 10 \\
			\hline
		    memory of CPNS-to-receiver channel for $\rho=0.95$ & $k_C$ & 15 \\
			\hline
			CPNS release time rate & $\lambda_e \tilde{T}$ & 0.001, 0.01, 0.1, 0.25 \text{and}\; 0.5 \\
			\hline
			CPNS release amplitude rate & $\lambda_a$ & $8\times 10^6, 10^5$, and $2\times 10^4$ \\
			\hline
			Receiver radius & $r_R$ & $0.5\; \mu m$ \\
			\hline
			Number of transmitted bits for PBS & $-$ & $10^8 \;$ bits \color{black}\\
			\hline
			Time step for PBS & $\Delta t $ & $10^{-3} \;$ $sec$ \\
			\hline
			
		\end{tabular}
	\end{center}
	\label{table12} 
\end{table}


\begin{figure}
	\centering
	\includegraphics[width=16 cm,height=5.5 cm]{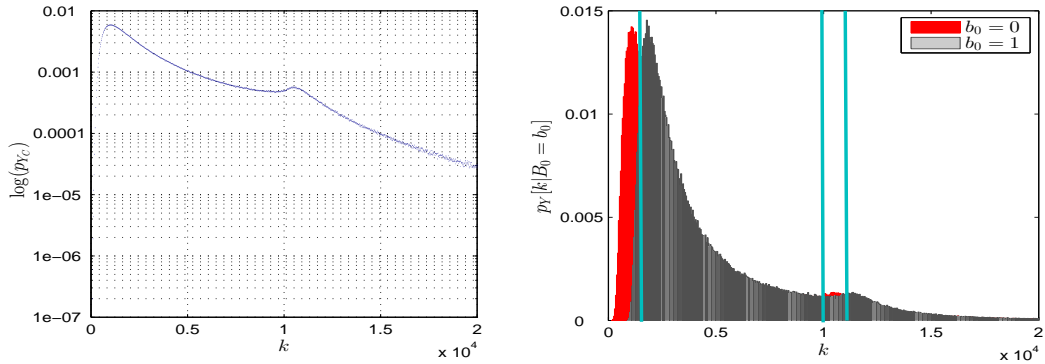}
	\setlength{\abovecaptionskip}{-0.6 cm}
	\caption{(Left) The Logarithm of CPNS noise distribution. (Right) The received signal distribution given transmission of bits 0 and 1. }
	\label{h3}
\end{figure}

\begin{figure}
	\centering
			\begin{minipage}{.5\linewidth}
	\includegraphics[width=8.5cm,height=6cm]{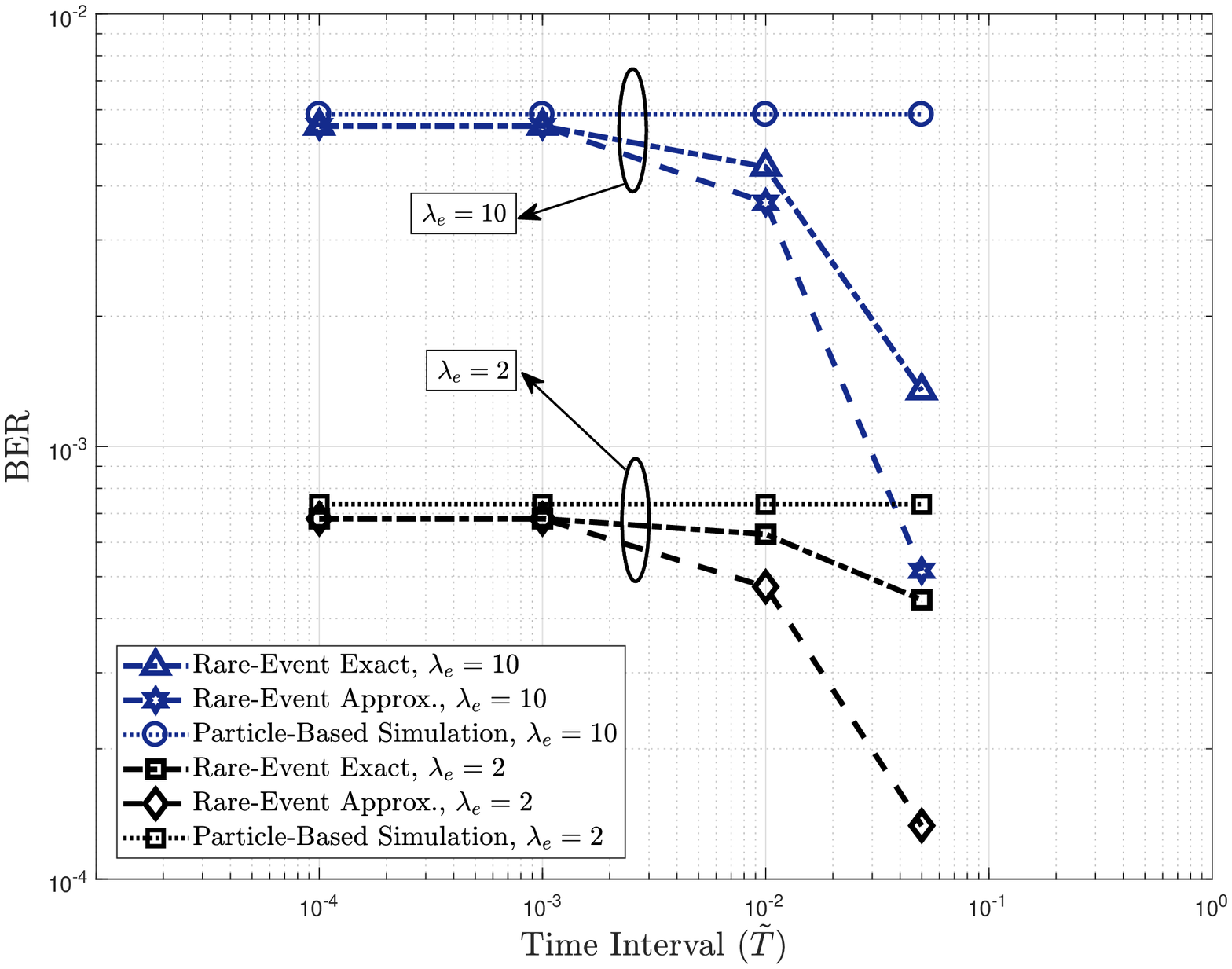}
		\setlength{\abovecaptionskip}{0 cm}
		\caption{Comparison of BERs (exact, approximation, and PBS) versus time interval $\tilde{T}$ for different values of  CPNS rate ($\lambda_e$).}
		\label{fig:comp}
	\end{minipage}\hspace{.25cm}
			\begin{minipage}{.45\linewidth}
		\centering
		\includegraphics[width=8.5cm,height=6cm]{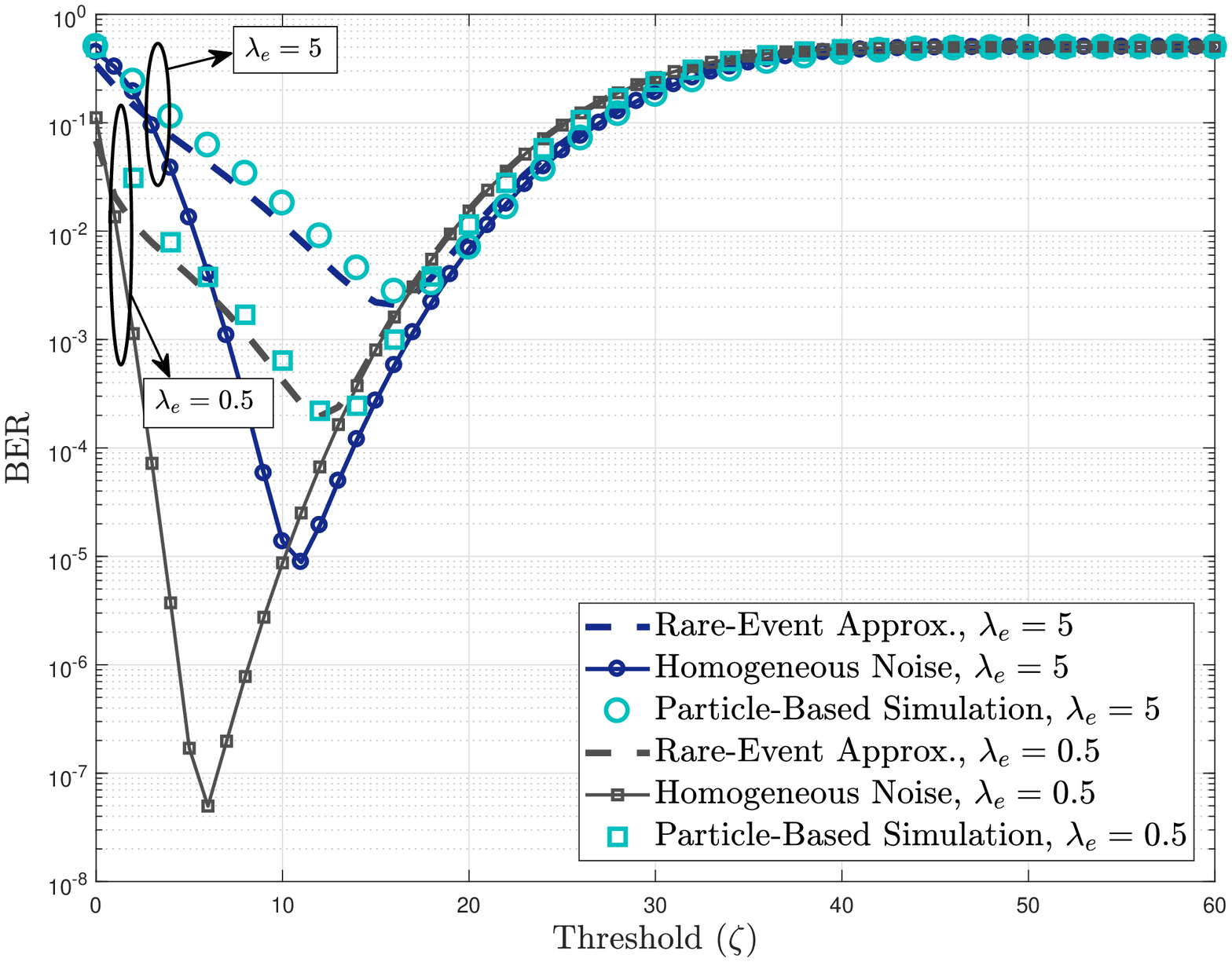}
		\setlength{\abovecaptionskip}{-0.5 cm}
		\caption{BER of the DMC system in the presence of a CPNS  obtained based on the rare-event approximation (confirmed by PBS) compared to the homogeneous Poisson noise assumption.}
		\label{f2}
	\end{minipage}
\end{figure}


\begin{figure}
	\centering
	\begin{minipage}{.5\linewidth}
	\includegraphics[width=8.5cm,height=6cm]{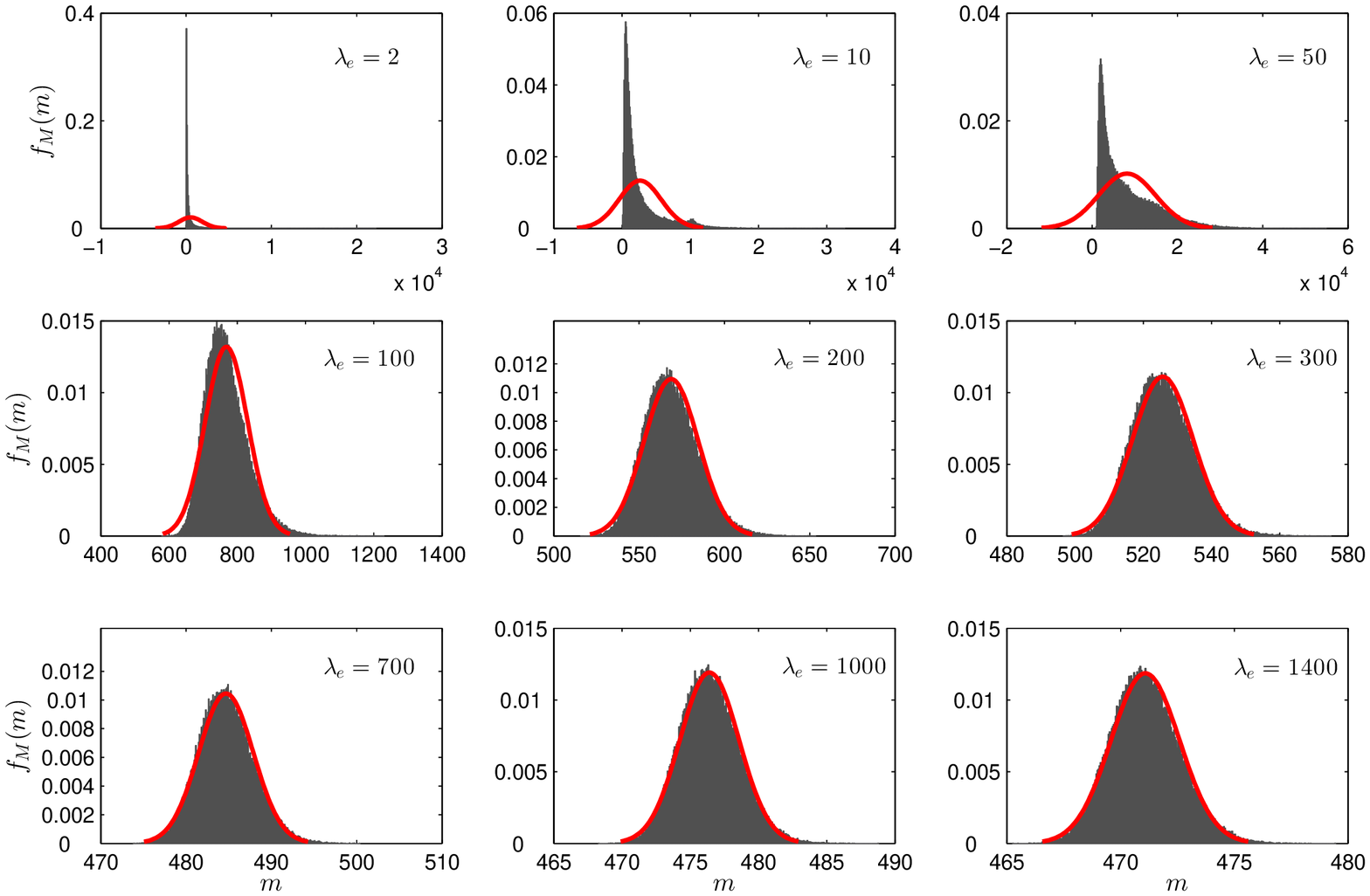}
		\setlength{\abovecaptionskip}{-.85 cm}
		\caption{Distribution of $M(t)$ in \eqref{fpg} obtained from simulation and normal approximation for various CPNS rate ($\lambda_e$).} 
		\label{h1}
	\end{minipage}\hspace{.25cm}
		\begin{minipage}{.45\linewidth}
		\centering
		\includegraphics[width=8.5cm,height=6cm]{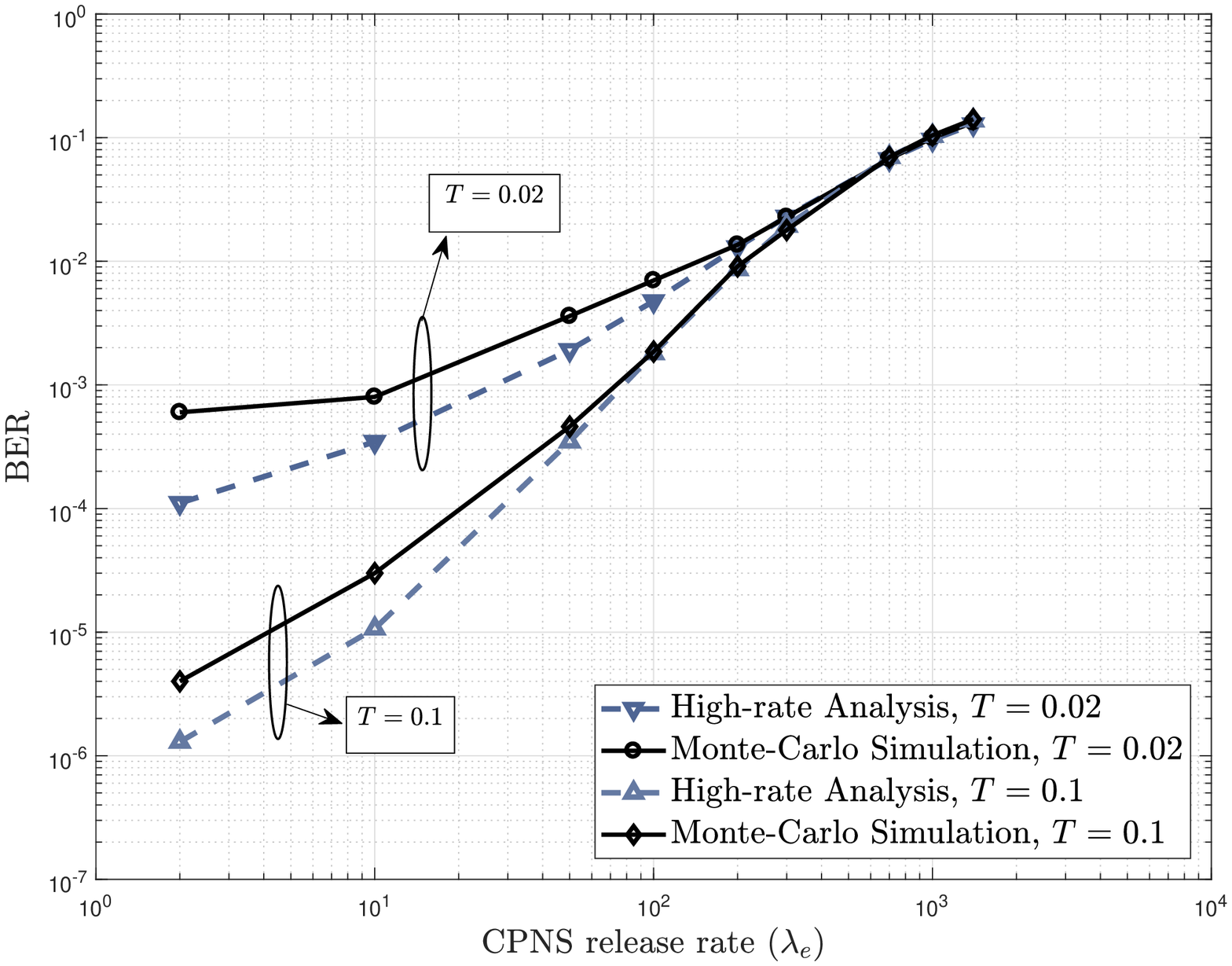}
		\setlength{\abovecaptionskip}{-0.5 cm}
		\caption{BER of the DMC system versus CPNS rate ($\lambda_e$) obtained based on simulation and high-rate analysis.}
		\label{h2}
	\end{minipage}
\end{figure}


\begin{figure}
	\centering
	\begin{minipage}{.5\linewidth}
	\includegraphics[width=8.5cm,height=5.5cm]{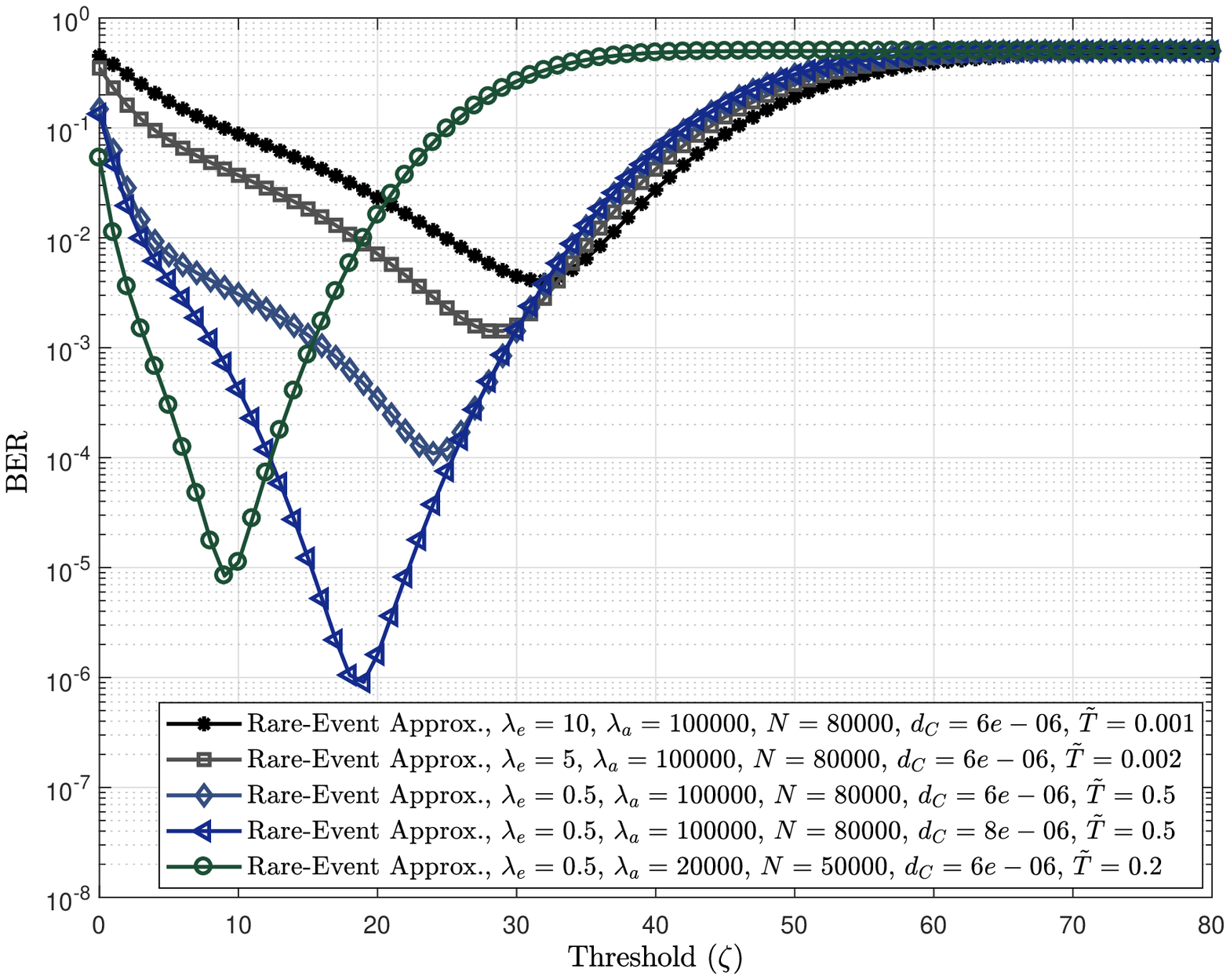}
		\setlength{\abovecaptionskip}{-0.5 cm}
		\caption{BER of the DMC system in the presence of a CPNS employing the rare-event approximation analysis.}
		\label{f4}
	\end{minipage}\hspace{.25cm}
				\begin{minipage}{.45\linewidth}
		\centering
		\includegraphics[width=8.5cm,height=5.5cm]{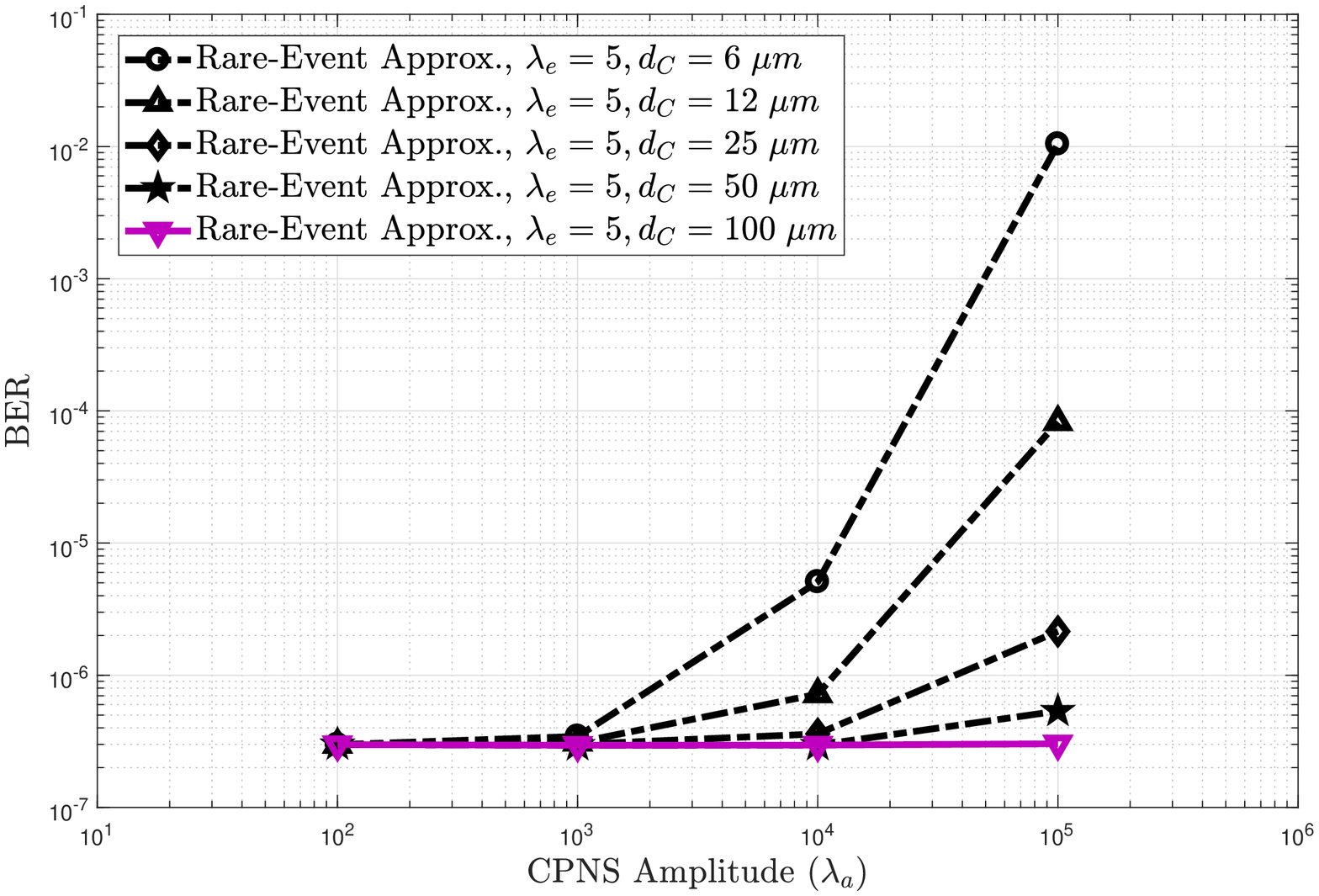}
       		\setlength{\abovecaptionskip}{-.5 cm}
		\caption{
BER of the DMC system in the presence of a CPNS versus CPNS ampitude ($\lambda_a$) for various distances.}	 
		\label{f8}
	\end{minipage}
\end{figure}



%

\color{black}

\appendices
\section{Equivalence of \eqref{conv1} and \eqref{conv}} \label{App-A}

It is straightforward to see that \eqref{conv1} and \eqref{conv} are equivalent. For simplicity of presentation, we show this equivalence for the special case of $\tilde{k}_C=3$. Starting from $\eqref{conv1}$, for $\tilde{k}_C=3$ we have:
\begin{equation}
\begin{aligned}
p_{Y_C}[k]=p_{\tilde{Y}_C^1}[k]\otimes p_{\tilde{Y}_C^2}[k] \otimes  p_{\tilde{Y}_C^{3}}[k],
\end{aligned}
\end{equation}
\normalsize
Substituting $p_{\tilde{Y}_C^1}[k]$, $p_{\tilde{Y}_C^2}[k]$, and $p_{\tilde{Y}_C^3}[k]$  with  \eqref{16} and employing $f_i[k]=p_{\tilde{Y}_C^i}[k|\mathcal{F}_1]$, we obtain: 
\small
\begin{equation}
\begin{aligned}
p_{Y_C}[k]&=((1-\lambda_e \tilde{T})\delta[k]+\lambda_e \tilde{T} {f}_1[k]) \otimes ((1-\lambda_e \tilde{T})\delta[k]+\lambda_e \tilde{T} {f}_2[k]) \otimes  ((1-\lambda_e \tilde{T})\delta[k]+\lambda_e \tilde{T} {f}_3[k]) \\
&=(1-\lambda_e \tilde{T})^3\delta[k]+(1-\lambda_e \tilde{T})^2(\lambda_e\tilde{T})({f}_1[k]+{f}_2[k]+{f}_3[k])+(1-\lambda_e \tilde{T})(\lambda_e \tilde{T})^2({f}_1[k] \otimes {f}_2[k]\\
&+{f}_1[k] \otimes {f}_3[k] + {f}_2[k] \otimes {f}_3[k])+(\lambda_e \tilde{T})^3({f}_1[k] \otimes {f}_2[k] \otimes {f}_3[k])\\
\end{aligned}
\end{equation}
\normalsize
which can be rewritten as follows:
\small
\begin{equation}
\begin{aligned}
p_{Y_C}[k]&=(1-\lambda_e \tilde{T})^3\delta[k]+(1-\lambda_e \tilde{T})^2(\lambda_e \tilde{T})\sum_{h=1}^{3}\delta[k] \otimes f_{\alpha_1^h}[k]\\
&+(1-\lambda_e \tilde{T})(\lambda_e \tilde{T})^2 \sum_{h=1}^{3}\delta[k] \otimes f_{\alpha_1^h}[k] \otimes f_{\alpha_2^h}[k]+(\lambda_e \tilde{T})^3 (\delta[k] \otimes f_{\alpha_1^1}[k] \otimes f_{\alpha_2^1}[k]\otimes f_{\alpha_3^1}[k])\\
&=\sum_{i=0}^{3}(1-\lambda_e \tilde{T})^{3-i}(\lambda_e \tilde{T})^i \sum_{h=1}^{\mathfrak{K}_i}  \delta[k] \otimes {f}_{\alpha_1^h}[k] \otimes \cdots \otimes {f}_{\alpha_i^h}[k].
\end{aligned}
\end{equation}
\normalsize
where $\mathfrak{K}_i  \overset{\Delta}{=}\left(\begin{array}{c}3\\ i\end{array}\right)$. This completes the proof.

\section{Deriving closed-form expression for $p_{Y_C}[k]$ in the high-rate regime} \label{dist HR2}
From \eqref{abc}, we have:
\begin{equation}  
\begin{aligned}
{p}_{Y_C}[k]&=\int_{0}^{+\infty} 
e^{-m}\frac{m^k}{k!}(2\pi k_2)^{-1/2}\mathrm{exp}\big(-\frac{(m-k_1)^2}{2k_2}\big)dm\\
&=\frac{(2\pi k_2)^{-1/2}}{k!}\int_{0}^{+\infty}m^k\mathrm{exp}\big(-m-\frac{(m-k_1)^2}{2k_2}\big)dm\\
&=\frac{(2\pi k_2)^{-1/2}}{k!}e^{-k_1^2/k_2}\int_{0}^{+\infty}m^k\mathrm{exp}\bigg(-\frac{m^2}{2k_2}-(1-\frac{k_1}{k_2})m\bigg)dm.
\end{aligned}	
\end{equation}
\normalsize
By changing variable $m=(2k_2)^{1/2}x$, we obtain
\small
\begin{equation}  
\begin{aligned} \label{bg}
{p}_{Y_C}[k]&=\frac{(2\pi k_2)^{-1/2}}{k!}e^{-k_1^2/k_2}({2k_2})^{(k+1)/2}\int_{0}^{+\infty}x^k\mathrm{exp}\bigg(-x^2-\sqrt{2k_2}(1-\frac{k_1}{k_2})x\bigg)dx.
\end{aligned}	
\end{equation}
\normalsize
On the other hand we have the following integral  \cite{NEW5}, \cite{NEW6}:
\small
\begin{equation}  
\begin{aligned} \label{aaa}
\int_{0}^{+\infty}x^{\nu}e^{-x^2-\gamma x}dx=2^{-(\nu+1)/2}\Gamma(\nu+1)e^{\gamma^2/8}D_{-\nu-1}(\gamma/\sqrt{2}).
\end{aligned}	
\end{equation}
\normalsize
Substituting $\gamma=\sqrt{2k_2}(1-\frac{k_1}{k_2})$ and $\nu=k$ into \eqref{aaa} and applying  \eqref{bg}, \eqref{distinf} is resulted.

\section{The proof for Lemma \ref{theorem2}} \label{lemma1}
Given a STD with the threshold value of $\zeta$, the BER of the system is given by \eqref{STDpe1}.
Therefore, we have
\small
\begin{equation} \label{derpe}
\begin{aligned}
\frac{\partial \mathrm{P}_e}{\partial \zeta}=\frac{1}{2}\Big(p_Y[\zeta|B_0=1]-p_Y[\zeta|B_0=0]\Big).
\end{aligned}
\end{equation} 
\normalsize

First we provide the direct proof, i.e., if  the  optimal ML detector in \eqref{ML22} is STD with optimal threshold $\zeta_o$, $\mathrm{P}_e$ in \eqref{STDpe1} is necessarily quasiconvex function of $\zeta$ with global minimum at $\zeta_o$: Assume that the optimal ML detector is a STD, where the optimal threshold value is $\zeta_o$. Considering \eqref{ML22} and \eqref{det}, we can write  
$1=\underset{b_0}{\operatorname{argmax}} \hspace{.25 cm} p_Y[\zeta|B_0=b_0],$ for all $\zeta>\zeta_o$ and $0=\underset{b_0}{\operatorname{argmax}} \hspace{.25 cm} p_Y[\zeta|B_0=b_0],$ for all $\zeta<\zeta_o$. Equivalently, we have $p_Y[\zeta|B_0=1]\geq p_Y[\zeta|B_0=0]$ and $p_Y[y|B_0=1]\leq p_Y[\zeta|B_0=0]$ for all $\zeta>\zeta_o$ and $\zeta<\zeta_o$, respectively. Therefore, we can write $p_Y[\zeta|B_0=1]-p_Y[\zeta|B_0=0]\geq 0$ and $p_Y[\zeta|B_0=1]-p_Y[\zeta|B_0=0]\leq 0$ for all $\zeta>\zeta_o$ and $\zeta<\zeta_o$, respectively. Regarding to \eqref{derpe}, $\partial \mathrm{P}_e/\partial \zeta$ is positive and negative for $\zeta>\zeta_o$ and $\zeta<\zeta_o$, respectively. Equivalently, $\mathrm{P}_e$ is decreasing function for $\zeta<\zeta_o$ and increasing for $\zeta>\zeta_o$. Thereby, $\mathrm{P}_e$ is a quasiconvex function of $\zeta$ which has a global minimum at $\zeta_o$. 

Now, we provide the converse proof, i.e., if BER given in \eqref{STDpe1} is quasiconvex fucntion of $\zeta$ with global minimum at $\zeta_o$, the optimal ML detector in \eqref{ML22} is a STD with optimal threshold $\zeta_o$: Assuming $\mathrm{P}_e$ given in \eqref{STDpe1} is quasiconvex of $\zeta$ with a global minimum at $\zeta_o$, $\partial \mathrm{P}_e/\partial \zeta$ is positive and negative for all $\zeta>\zeta_o$ and $\zeta<\zeta_o$, respectively. Therefore, we have $p_Y[\zeta|B_0=1]-p_Y[\zeta|B_0=0]\geq0$ and $p_Y[\zeta|B_0=1]-p_Y[\zeta|B_0=0]\leq0$ for $\zeta>\zeta_o$ and $\zeta<\zeta_o$, respectively, considering \eqref{derpe}. Equivalently, $p_Y[\zeta|B_0=1]\geq p_Y[y|B_0=0]$ and $p_Y[\zeta|B_0=1]\leq p_Y[y|B_0=0]$ for all $\zeta>\zeta_o$ and $\zeta<\zeta_o$, respectively. Therefore, we can write  
$0=\underset{b_0}{\operatorname{argmax}} \hspace{.25 cm} p_Y[\zeta|B_0=b_0],$ for $\zeta<\zeta_o$ and $1=\underset{b_0}{\operatorname{argmax}} \hspace{.25 cm} p_Y[\zeta|B_0=b_0],$ for $\zeta>\zeta_o$ which is equivalent to optimality of STD with threshold $\zeta_o$, considering \eqref{ML22} and \eqref{det}.
\ifCLASSOPTIONcaptionsoff
  \newpage
\fi

\end{document}